\documentclass[a4paper,UKenglish,numberwithinsect,cleveref,thm-restate]{lipics-v2021}
\hypersetup{hypertexnames=false} 
\usepackage{mathpartir,upgreek,mathbbol,tikz}

\newcommand{\pr}{\mathsf{pr}}
\newcommand{\np}{\mathsf{np}}
\newcommand{\otyp}{\mathsf{o}}
\renewcommand{\r}{\mathsf{r}}
\newcommand{\Pp}{\mathcal{P}}
\newcommand{\Gg}{\mathcal{G}}

\newcommand{\Tt}{\mathcal{T}}
\newcommand{\Nn}{\mathcal{N}}
\newcommand{\Rr}{\mathcal{R}}

\newcommand{\restr}{{\restriction}}
\newcommand{\zero}{\mathbf{0}}
\newcommand{\Nat}{\mathbb{N}}
\newcommand{\ord}{\mathit{ord}}
\newcommand{\dom}{\mathit{dom}}

\newcommand{\BT}{\mathit{BT}}

\newcommand{\leta}{\mathsf{a}}
\newcommand{\letb}{\mathsf{b}}
\newcommand{\letc}{\mathsf{c}}
\newcommand{\p}{\mathfrak{p}}
\newcommand{\setP}{\Sigma_\leta}

\newcommand{\arr}{\mathbin{\to}}

\newcommand{\dupl}{\mathit{dupl}}
\newcommand{\redb}{\to_\beta}
\newcommand{\lamdots}{.\cdots{}.}
\newcommand{\sort}{\mathsf{sort}}
\newcommand{\floor}[1]{\lfloor #1\rfloor}
\newcommand{\set}[1]{\{#1\}}
\newcommand{\setof}[2]{\{#1\mid#2\}}

\nolinenumbers

\Crefname{equation}{Equality}{Equalities}
\Crefname{equalities}{Equalities}{Equalities}\creflabelformat{equalities}{(#2#1#3)}
\Crefname{inequality}{Inequality}{Inequalities}\creflabelformat{inequality}{(#2#1#3)}
\Crefname{implication}{Implication}{Implications}\creflabelformat{implication}{(#2#1#3)}

\EventEditors{Miko{\l}aj Boja\'{n}czyk, Emanuela Merelli, and David P. Woodruff}
\EventNoEds{3}
\EventLongTitle{49th International Colloquium on Automata, Languages, and Programming (ICALP 2022)}
\EventShortTitle{ICALP 2022}
\EventAcronym{ICALP}
\EventYear{2022}
\EventDate{July 4--8, 2022}
\EventLocation{Paris, France}
\EventLogo{}
\SeriesVolume{229}
\ArticleNo{117}

\title{Unboundedness for Recursion Schemes: A~Simpler Type System}

\author{David Barozzini}{Institute of Informatics, University of Warsaw, Poland}{dbarozzini@mimuw.edu.pl}{}{}
\author{Paweł Parys}{Institute of Informatics, University of Warsaw, Poland}{parys@mimuw.edu.pl}{https://orcid.org/0000-0001-7247-1408}{}
\author{Jan Wróblewski}{Institute of Informatics, University of Warsaw, Poland}{xi@mimuw.edu.pl}{https://orcid.org/0000-0002-0229-4239}{}

\authorrunning{D. Barozzini, P. Parys, and J. Wróblewski}

\Copyright{David Barozzini, Paweł Parys, and Jan Wróblewski}

\ccsdesc[500]{Mathematics of computing~Lambda calculus}
\ccsdesc[300]{Theory of computation~Verification by model checking}
\ccsdesc[100]{Theory of computation~Rewrite systems}

\keywords{Higher-order recursion schemes, boundedness, intersection types, safe schemes}

\relatedversion{This is an extended version of the following paper:}
\relatedversiondetails{Published version}{https://doi.org/10.4230/LIPIcs.ICALP.2022.117}

\supplement{Implementation of the presented algorithm is available at \url{https://github.com/xi-business/infsat}.}

\funding{Work supported by the National Science Centre, Poland (grant no.\@ 2016/22/E/ST6/00041).}

\begin{document}

\maketitle

\begin{abstract}
	Decidability of the problems of unboundedness and simultaneous unboundedness (aka.\@ the diagonal problem) for higher-order recursion schemes
	was established by Clemente, Parys, Salvati, and Walukiewicz (2016).
	Then a procedure of optimal complexity was presented by Parys (2017); this procedure used a complicated type system, involving multiple flags and markers.
	We present here a simpler and much more intuitive type system serving the same purpose.
	We prove that this type system allows to solve the unboundedness problem for a widely considered subclass of recursion schemes, called safe schemes.
	For unsafe recursion schemes we only have soundness of the type system: if one can establish a type derivation claiming that a recursion scheme is unbounded then it is indeed unbounded.
	Completeness of the type system for unsafe recursion schemes is left as an open question.
	Going further, we discuss an extension of the type system that allows to handle the simultaneous unboundedness problem.

	We also design and implement an algorithm that fully automatically checks unboundedness of a given recursion scheme, completing in a short time for a wide variety of inputs.
\end{abstract}

\section{Introduction}

	\emph{Higher-order recursion schemes} (\emph{recursion schemes} for short) proved to be useful in model-checking programs using higher-order functions,
	see e.g.\@ Kobayashi~\cite{Kobayashi-jacm}
	(recursion schemes are algorithmically manageable abstractions of such programs, faithfully representing the control flow).
	Higher-order functions are widely used in functional programming languages, like Haskell, OCaml, and Lisp;
	additionally, higher-order features are now present in most mainstream languages like Java, JavaScript, Python, or C++.

	The formalism of recursion schemes is equivalent via direct translations to simply-typed $\lambda Y$-calculus~\cite{schemes-lY} and to higher-order OI grammars~\cite{Damm82,KobeleSalvati}.
	Collapsible pushdown systems~\cite{collapsible} and ordered tree-pushdown systems~\cite{DBLP:conf/fsttcs/ClementePSW15} are other equivalent formalisms.
	Recursion schemes cover some other models such as indexed grammars~\cite{Aho68} and ordered multi-pushdown automata~\cite{OrderedMultiPushdown}.

	The usefulness of recursion schemes follows from the fact that trees generated by them have decidable MSO theory~\cite{ong-lics}.
	When the property to be verified is given by a parity automaton (a formalism equivalent to the MSO logic), and when the recursion scheme is of order $m$,
	the model-checking problem is $m$-\textsf{EXPTIME}-complete; already for reachability properties the problem is $(m-1)$-\textsf{EXPTIME}-complete~\cite{complexity-hors}.
	Although this high complexity may be threatening, there exist algorithms that behave well in practice.
	They make use of appropriate systems of intersection types.
	Namely, a Japanese group created model-checkers \textsc{TRecS}~\cite{Kobayashi-jacm} and \textsc{HorSat}~\cite{horsat},
	and prototype verification tools, \textsc{MoCHi}~\cite{mochi} and \textsc{EHMTT} verifier~\cite{ehmtt}, on top of them.
	A group hosted in Oxford created model-checkers \textsc{HORSC}~\cite{HORSC} and \textsc{TravMC2}~\cite{TravMC2}.
	Necessarily these model-checkers are very slow on some worst-case examples,
	but on schemes generated from some real life higher-order programs they usually work in a reasonable time
	(while in the worst-case a huge type derivation may be required, it is often the case that there exists a small type derivation that can be found quickly).

	In recent years, interest has arisen in model checking recursion schemes against properties that are not regular (i.e., not expressible in the MSO logic).
	This primarily concerns the \emph{unboundedness problem} for word languages recognized by recursion schemes~\cite{parys-itrs},
	and its generalization---the \emph{simultaneous unboundedness problem} (aka.\@ the diagonal problem)~\cite{diagonal-safe,diagonal,diagonal-types}.
	Decidability of the latter problem implied computability of downward closures (with respect to the subsequence relation) for these languages~\cite{Zetzsche-downward-closure},
	and decidability of their separability by piecewise testable languages~\cite{piecewise-testable}.
	It was also possible to establish decidability of the WMSO+$\mathsf{U}$ logic (an extension of MSO with a quantifier $\mathsf{U}$, talking about unboundedness)
	over trees generated by recursion schemes~\cite{wmsou-schemes}.

	Moreover, there is also a link to asynchronous shared-memory programs, being a common way to manage concurrent requests in a system.
	In asynchronous programming,
	each asynchronous function is a sequential program.
	When run, it can change the global shared state of the program and run other asynchronous functions.
	A scheduler repeatedly and non-deterministically executes pending asynchronous functions.
	Majumdar, Thinniyam and Zetzsche~\cite{asynchronous} have proven that when asynchronous functions are modeled as recursion schemes, the question
	whether there is an a priori upper bound on the number of pending executions of asynchronous functions
	can be reduced to the problem we consider here, namely the unboundedness problem for a single recursion scheme.

	In this paper we revisit the (simultaneous) unboundedness problem for word languages recognized by recursion schemes.
	This problem asks, for a set of letters $\setP$ and a language of words $L$ (recognized by a recursion scheme),
	whether for every $n\in\Nat$ there is a word in $L$ where every letter from $\setP$ occurs at least $n$ times.
	Equivalently, one can consider a recursion scheme generating an infinite tree $t$
	and ask whether for every $n\in\Nat$ there is a finite branch in $t$ where every letter from $\setP$ occurs at least $n$ times.
	The problem is already interesting when $|\setP|=1$; then we talk about the unboundedness problem.
	Decidability of the simultaneous unboundedness problem was first established by Hague, Kochems, and Ong~\cite{diagonal-safe},
	for a well-recognized subclass of recursion schemes, called \emph{safe} schemes~\cite{Damm82,easy-trees,blum-ong-safe,selection,schemes-lY}.
	The solution was then generalized to all recursion schemes by Clemente, Parys, Salvati, and Walukiewicz~\cite{diagonal}.
	These two algorithms are useless in practice, not only because their complexity is much higher than the optimal one,
	but mainly because they perform some transformations of recursion schemes that are extremely costly in every case, not only in the worst case.
	The complexity of the problem for recursion schemes (namely, $(m-1)$-\textsf{EXPTIME}-completeness for schemes of order $m$) was settled by Parys~\cite{diagonal-types};
	moreover, his solution uses intersection types, and thus it is potentially suitable for implementation
	(by analogy to the regular model-checking case, where algorithms using type systems led to reasonable implementations).
	In the single-letter case ($|\setP|=1$) the problem is still $(m-1)$-\textsf{EXPTIME}-complete;
	a slightly simpler type system for this case was presented by Parys~\cite{parys-itrs}.
	Unfortunately, the type systems of Parys~\cite{parys-itrs,diagonal-types} have two drawbacks.
	First, they are quite complicated: type judgments can be labeled by different kinds of flags and markers, which influence type derivations in a convoluted way.
	Second (and related to first), it seems that the number of choices for these flags and markers is quite large,
	and thus finding type derivations even for quite simple recursions schemes may be very costly.
	Having these drawbacks in mind, we leave experimental evaluation of these algorithms for future work.

	In this paper we rather consider a much simpler type system, proposed by Parys in his survey~\cite{parys-survey},
	based on an earlier work concerning lambda-terms representing functions on numerals~\cite{jfp-numerals}.
	This type system is much easier than the previous one: every type is labeled only by a single ``productivity flag'' having an intuitive meaning.
	Namely, the flag says whether the lambda-term under consideration is responsible for creating occurrences of the letter from $\setP$.
	It was only conjectured that this type system may be used to solve the unboundedness problem.

	Our contributions are as follows:
	\begin{itemize}
	\item	We prove that the type system allows to solve the unboundedness problem for all safe recursion schemes.
	\item	We show that the algorithm using the type system solves the unboundedness problem for safe recursion schemes of order $m$ in $(m-1)$-\textsf{EXPTIME} (being optimal).
	\item	We prove soundness of the type system for all (i.e., also unsafe) recursion schemes,
		saying that if one has found a type derivation claiming that a recursion schemes is unbounded then it is indeed unbounded.
		Completeness of the type system for unsafe recursion schemes is left as an open problem.
	\item	We implement an algorithm solving the unboundedness problem by means of the proposed type system, and we present results of our experiments.
		The outcome of the algorithm is always correct if the recursion scheme given on input is safe.
		When the recursion scheme is not safe, proofs of its unboundedness found by the algorithm are still guaranteed to be correct.
		However, our theorems do not guarantee that the algorithm will find a proof of unboundedness in the unsafe case, so if the algorithm fails to find such a proof,
		it does not mean that the unsafe recursive scheme is necessarily bounded.
	\item	We then generalize the type system to the simultaneous unboundedness problem (i.e., the multiletter case), obtaining the same properties:
		algorithm in $(m-1)$-\textsf{EXPTIME} for safe recursion schemes of order $m$, and soundness for all recursion schemes.
	\item	We implement an optimized version of our algorithm, \textsc{InfSat}, which is fast for a wide variety of inputs. We build upon the implementation of \textsc{HorSat}~\cite{horsat}.
		We modify \textsc{HorSat} benchmarks to conform to \textsc{InfSat} input format,
		and present their results to show that \textsc{InfSat} is able to handle inputs described by Broadbent and Kobayashi~\cite{horsat} as ``practical'' as well as additional
		benchmarks crafted to measure the speed of \textsc{InfSat}.
	\end{itemize}

\section{Preliminaries}\label{sec:prelim}

	The set of \emph{sorts} is constructed from a unique basic sort $\otyp$ using a binary operation $\arr$.
	Thus $\otyp$ is a sort and if $\alpha,\beta$ are sorts, so is $(\alpha\arr\beta)$.
	The order of a sort is defined by: $\ord(\otyp)=0$, and $\ord(\alpha\arr\beta)=\max(1+\ord(\alpha),\ord(\beta))$.

	A \emph{signature} $\Sigma$ is a set of typed constants, that is, symbols with associated sorts.
	We assume that sorts of all constants in the signature are of order at most $1$, that is, of the form $\underbrace{\otyp\arr\dots\arr\otyp\arr{}}_r\otyp$;
	for a constant of such a sort, $r$ is called its \emph{arity}.
	We fix a distinguished constant $\upomega\in\Sigma$ of arity $0$ (it will be used in places where a computation diverges).

	The set of \emph{infinitary simply-typed lambda-terms} is defined coinductively as follows.
	A constant $a\in\Sigma$ of sort $\alpha$ is a lambda-term of sort $\alpha$.
	For each sort $\alpha$ there is a countable set of variables $x^\alpha,y^\alpha,\dots$ that are also lambda-terms of sort $\alpha$.
	If $M$ is a lambda-term of sort $\beta$ and $x^\alpha$ a variable of sort $\alpha$ then $\lambda x^\alpha.M$ is a lambda-term of sort $\alpha\arr\beta$.
	Finally, if $M$ is a lambda-term of sort $\alpha\arr\beta$ and $N$ is a lambda-term of sort $\alpha$ then $M\,N$ is a lambda-term of sort $\beta$.
	As usual, we identify lambda-terms up to alpha-conversion.
	We often omit the sort annotation of variables, but formally every variable has a sort.
	We use the standard notions of free variables, substitution, and beta-reduction
	(of course during substitution and beta-reduction we rename bound variables to avoid name conflicts).
	A lambda-term is called \emph{closed} when it does not have free variables.
	For a lambda-term $M$ of sort $\alpha$, the order of $M$, denoted $\ord(M)$, is defined as $\ord(\alpha)$.

	The \emph{complexity} of a lambda-term $M$ is the maximum of orders of those subterms of $M$ that are not of the form $a\,M_1\,\dots\,M_k$, where $a$ is a constant and $k\geq 0$
	(or $0$ if there are no such subterms).
	Note that for most lambda-terms $M$ the complexity is just the maximum of orders of all subterms of $M$;
	the difference is only at complexity $0$ and $1$: we want lambda-terms built entirely from constants to have complexity $0$, not $1$.
	
	A closed lambda-term of sort $\otyp$ and complexity $0$ is called a \emph{tree}.
	Equivalently, a lambda-term is a tree if it is of the form $a\,M_1\,\dots\,M_r$, where $M_1,\dots,M_r$ are trees, and $r$ is the arity of $a$.
	While talking about a \emph{branch} of a tree, we mean a finite branch that ends in a leaf not being $\upomega$-labeled.
	Formally, a \emph{(finite) branch} of a tree $T=a\,T_1\,\dots\,T_r$ is a sequence of constants $a_1,a_2,\dots,a_k$ such that $a_1=a\neq\upomega$,
	and either $a_2,\dots,a_k$ (with $k\geq 2$) is a finite branch of some $T_i$, or $r=0$ and $k=1$.

	We consider B\"ohm trees only for closed lambda-terms of sort $\otyp$.
	For such a lambda-term $M$, its \emph{B\"ohm tree} is constructed by coinduction, as follows:
	if there is a sequence of beta-reductions from $M$ to a lambda-term of the form $a\,M_1\,\dots\,M_r$,
	and $T_1,\dots,T_r$ are B\"ohm trees of $M_1,\dots,M_r$, respectively, then $a\,T_1\,\dots\,T_r$ is a B\"ohm tree of $M$;
	if there is no such a sequence of beta-reductions from $M$, then the constant $\upomega$ is a B\"ohm tree of $M$.
	It is folklore that every closed lambda-term of sort $\otyp$ has exactly one B\"ohm tree (the order in which beta-reductions are performed does not matter); this tree is denoted by $\BT(M)$.
	Notice that a B\"ohm tree is indeed a tree, and that if $M$ is finite then its B\"ohm tree equals the beta-normal form of $M$.

	A \emph{recursion scheme} $\Gg$ is a finite representation of a closed lambda-term $\Lambda(\Gg)$ that is of sort $\otyp$ and regular,
	that is, has finitely many different subterms.
	We postpone the definition of a recursion scheme until \cref{sec:algorithm}.
	As the \emph{order} of $\Gg$ we understand the complexity of $\Lambda(\Gg)$.
	We do not claim that every regular lambda-term of sort $\otyp$ can be directly represented by a recursion scheme,
	however we remark that every regular lambda-term of sort $\otyp$ is equivalent (in the sense of having the same B\"ohm tree) to some lambda-term represented by a recursion scheme.

\subparagraph{Simultaneous unboundedness problem.}

	Fix a set of \emph{important constants} $\setP$.
	We say that a closed lambda-term $M$ of sort $\otyp$ is \emph{unbounded} (with respect to $\setP$)
	if for every $n\in\Nat$ there exists a finite branch of $\BT(M)$ with at least $n$ occurrences of every constant from $\setP$.
	The \emph{simultaneous unboundedness problem} (SUP) is to decide, given a recursion scheme $\Gg$, whether $\Lambda(\Gg)$ is unbounded.

	In the special case of $|\setP|=1$ we talk about the \emph{unboundedness problem}.

\section{Type system}\label{sec:types}

	In this section we present a type system that allows us to solve the (single-letter) unboundedness problem.
	The type system was first proposed in Parys' survey~\cite{parys-survey}, without any correctness proofs.

	In the remaining part of the paper, except for \cref{sec:multi-letter}, we assume that the signature $\Sigma$ consists of four constants: $\leta$ of arity $1$, $\letb$ of arity $2$, and $\letc$ and $\upomega$ of arity $0$,
	where only $\leta$ is important, that is $\setP=\set{\leta}$.
	This is without loss of generality, since we can replace a constant of arbitrary arity by a combination of these four constants, without changing the answer to the unboundedness problem.
	
	Before defining the type system (and safety), let us state a theorem describing its desired properties:
	
	\begin{theorem}\label{thm:main}
		The following two statements are equivalent for every safe closed lambda-term $M$ of sort $\otyp$:
		\begin{enumerate}[(1)]
		\item	$M$ is unbounded (i.e., for every $n\in\Nat$ there exists a finite branch of $\BT(M)$ with at least $n$ occurrences of the constant $\leta$);
		\item	for every $n\in\Nat$ there exists $v\geq n$ such that one can derive $\emptyset\vdash M:(v,\r)$.
		\end{enumerate}
	\end{theorem}
	
	In this theorem, type judgments contain a natural number $v$, called a \emph{productivity value}.
	The goal of this value is to approximately count the number of occurrences of the important constant $\leta$ on a selected branch of $\BT(M)$.

	Types in the type system differ from sorts in that on the left side of $\arr$, instead of a single type, we have a set of so-called \emph{type pairs} $(f,\tau)$,
	where $\tau$ is a \emph{type}, and $f\in\set{\pr,\np}$ is a \emph{productivity flag} (where $\pr$ stands for productive, and $\np$ for nonproductive).
	The unique atomic type is denoted $\r$.
	More precisely, for each sort $\alpha$ we define the set $\Tt^\alpha$ of types of sort $\alpha$ as follows:
	\begin{align*}
		\Tt^\otyp=\set{\r},&&
		\Tt^{\alpha\arr\beta}=\Pp(\set{\pr,\np}\times\Tt^\alpha)\times\Tt^\beta,
	\end{align*}
	where $\Pp$ denotes the powerset.
	A type $(T,\tau)\in\Tt^{\alpha\arr\beta}$ is denoted as $\bigwedge T\arr\tau$,
	or $\bigwedge_{i\in I}(f_i,\tau_i)\arr\tau$ when $T=\setof{(f_i,\tau_i)}{i\in I}$.
	In this notation we implicitly assume that all the pairs $(f_i,\tau_i)$ are different.
	The empty intersection is denoted by $\top$.
	Moreover, to our terms we will not only assign a type $\tau$, but also a productivity flag $f\in\set{\pr,\np}$ (which together form a pair $(f,\tau)$).
	Let us emphasize that for every sort $\alpha$ the set $\Tt^\alpha$ is finite.

	Intuitively, a lambda-term has type $\bigwedge T\arr\tau$ when it can return $\tau$, while taking an argument for which we can derive all type pairs from $T$;
	simultaneously, while having such a type, the lambda-term is obligated to use its arguments in all ways described by type pairs from $T$.
	For example, the lambda-term $\lambda x.\letc$ does not use its argument, and hence it is necessarily of type $\top\arr\r$ (i.e., $\bigwedge T\arr\tau$ with $T=\emptyset$).
	
	To determine the productivity flag $f$ assigned to a lambda-term $M$, we should imagine that $M$ is a subterm of a closed term $K$ of sort $\otyp$,
	and we should select some finite branch in $\BT(K)$ (with the intuition that different choices of the branch correspond to different type derivations).
	Then, we assign to $M$ the flag $\pr$ (productive) when the subterm $M$ is responsible for increasing the number of occurrences of the constant $\leta$ on the selected branch.
	To be more precise, a lambda-term is responsible for producing occurrences of a constant $\leta$ in two cases.
	First, when it explicitly contains the constant $\leta$---assuming that this $\leta$ will be placed on the selected branch.
	Second, when it takes a productive argument (i.e., an argument responsible for producing $\leta$) and uses it at least twice.
	The first possibility occurs for example in the lambda-term $M_1=\lambda x.\leta\,x$; the constant $\leta$ is explicitly produced.
	In order to see the second possibility, consider the lambda-term $M_2=\lambda y.\lambda x.y\,(y\,x)$,
	and suppose that the argument received for $y$ is the aforementioned productive lambda-term $\lambda x.\leta\,x$, which outputs the constant $\leta$.
	Then, the lambda-term $M_2$ is itself responsible for increasing the number of occurrences of $\leta$ in the resulting tree,
	compared to the number of occurrences produced by the argument.
	Notice that the same lambda-term $M_2$ has also another type: the argument received for $y$ may be nonproductive (say $\lambda x.x$),
	and then $M_2$ becomes nonproductive as well.
	Next, let us compare $M_2$ with the lambda-term $M_2'=\lambda y.\lambda x.y\,x$, when used with the argument $\lambda x.\leta\,x$.
	Although the argument is used, and one $\leta$ is output, the lambda-term $M_2'$ has nothing to do with increasing the number of occurrences of $\leta$; it is nonproductive.

	A \emph{type judgment} is of the form $\Gamma\vdash M:(v,\tau)$, where we require that the type $\tau$ and the lambda-term $M$ are of the same sort.
	The \emph{type environment} $\Gamma$ is a set of bindings of variables of the form $x^\alpha:(g,\tau)$, where $g\in\set{\pr,\np}$ and $\tau\in\Tt^\alpha$.
	In $\Gamma$ we may have multiple bindings for the same variable.
	By $\dom(\Gamma)$ we denote the set of variables $x$ which are bound by $\Gamma$, and
	by $\Gamma\restr_\pr$ we denote the set of only those bindings $x:(g,\tau)$ from $\Gamma$ in which $g=\pr$.
	We are not only interested in whether some type can be derived, but we also want to assign a value to every derivation;
	to this end, a type judgment contains a number $v\in\Nat$, which is called a \emph{productivity value}.
	Having $v=0$ corresponds to the $\np$ flag, while positive values correspond to the $\pr$ flag.
	Thus, while a productivity flag says only whether any occurrence of the important constant $\leta$ is produced,
	the productivity value approximates (is a lower bound on) the number of produced occurrences of this constant.
	Note that in type judgments we store the productivity value, coming from the infinite set $\Nat$, while in types we abstract this value to the productivity flag,
	which allows us to have finitely many types.

	We now present rules of the type system, starting from rules for constants:
	\begin{mathpar}
	\inferrule{}{
		\emptyset\vdash \leta:(1,(f,\r)\arr\r)
	}\and
	\inferrule{}{
		\emptyset\vdash \letc:(0,\r)
	}\\
	\inferrule{}{
		\emptyset\vdash \letb:(0,(f,\r)\arr\top\arr\r)
	}\and
	\inferrule{}{
		\emptyset\vdash \letb:(0,\top\arr(f,\r)\arr\r)
	}
	\end{mathpar}
	Notice that in the rule for $\letb$ we have a type $(f,\r)$ (with an arbitrary flag $f$) only for one of the two arguments;
	this corresponds to the fact that we are interested in a single branch of the B\"ohm tree, so we want to descend only to a single child.
	Moreover, we do not have a rule for the constant $\upomega$, because by definition a branch cannot contain occurrences of this constant.

	While typing a variable $x$, we take its type from the type environment, and we use $0$ as the productivity value.
	The lambda-term $x$ itself is not responsible for producing any constants, no matter whether the lambda-term substituted for $x$ will produce any constants or not.
	A lambda-term becomes productive when a productive variable $x$ is used twice; we account for that in the application typing rule $(@)$.
	\begin{mathpar}
	\inferrule{}{
		x:(f,\tau)\vdash x:(0,\tau)
	}
	\end{mathpar}

	When we pass through a lambda-binder, we simply move some type pairs between the argument and the type environment:
	\begin{mathpar}
	\inferrule*[right=($\lambda$)]{
		\Gamma\cup\setof{x:(f_i,\tau_i)}{i\in I}\vdash K:(v,\tau)
	\\
		x\not\in \dom(\Gamma)
	}{
		\Gamma\vdash\lambda x.K:(v,\textstyle\bigwedge_{i\in I}(f_i,\tau_i)\arr\tau)
	}
	\end{mathpar}

	Before giving the last rule, we need one more definition.
	Given a family of type environments $\left(\Gamma_i\right)_{i \in J}$, a \emph{duplication factor}, denoted $\dupl((\Gamma_i)_{i\in J})$, equals
	$\sum_{i\in J}\left|\Gamma_i\restr_\pr\right|-\left|\bigcup_{i\in J}\Gamma_i\restr_\pr\right|$.
	It counts the number of repetitions (``duplications'') of productive type bindings in the type environments:
	a productive type binding belonging to one type environment does not add anything, a productive type binding belonging to two type environments adds $1$, and so on.
	\begin{mathpar}
	\inferrule*[right=$(@)$]{
		0\not\in I
	\\
		\forall i\in I.\ (f_i=\pr)\Leftrightarrow (v_i>0 \lor \Gamma_i\restr_\pr\neq\emptyset)
	\\
		\Gamma_0\vdash K:(v_0,\textstyle\bigwedge_{i\in I}(f_i,\tau_i)\arr\tau)
	\\
		\Gamma_i\vdash L:(v_i,\tau_i)\mbox{ for each }i\in I
	}{
		\textstyle\bigcup_{i\in\set{0}\cup I}\Gamma_i\vdash K\,L:(\dupl((\Gamma_i)_{i\in\set{0}\cup I})+\textstyle\sum_{i\in\set{0}\cup I}v_i,\tau)
	}
	\end{mathpar}
	Here, by using the notation $\bigwedge_{i\in I}(f_i,\tau_i)$, we assume that the pairs $(f_i,\tau_i)$ are all different.

	In the rule above, the condition $(f_i=\pr)\Leftrightarrow (v_i>0 \lor \Gamma_i\restr_\pr\neq\emptyset)$ means that when $K$ requires a ``productive'' argument,
	then either we can apply an argument $L$ that is itself productive,
	or we can apply a nonproductive $L$ that uses a productive variable (the argument obtained after substituting something for this variable will become productive).

	Using the $\dupl$ function we realize the intuition that when a variable responsible for creating occurrences of $\leta$ (i.e., productive) is used at least twice,
	then the lambda-term is itself responsible for increasing the number of occurrences of $\leta$;
	thus we add $\dupl$ to the productivity value.

	Because we are interested in counting duplications of type bindings, it was necessary to require that every type binding from the type environment is actually used somewhere
	(in particular $\Gamma\vdash M:(f,\tau)$ does not necessarily imply $\Gamma,x:(g,\sigma)\vdash M:(f,\tau)$).
	On the other hand, a type environment is a set: a repeated usage of a type binding is counted in the productivity value, but is not reflected in the type environment.

	Let us underline that although we consider infinite lambda-terms, all type derivations are required to be finite.
	This may look suspicious, but note that when a lambda-term has type $\top\arr\tau$ (like, e.g., the constant $\letb$), then we need no derivation for its argument.

	\begin{example}\label{ex:1}
		Let us give an example of a derivation for a lambda-term $\lambda y. \lambda z. y\,(y\,(\leta\,z))$ of sort $(\otyp \arr \otyp) \arr \otyp \arr \otyp$.
		In this particular derivation, we will assume that $y$ and $z$ are both productive.
		\begin{mathpar}
		\inferrule*[right=$(@)$]{
			\vdash \leta : (1, (\pr, \r) \arr \r) \\
			z:(\pr,\r)\vdash z:(0,\r)
		}{z:(\pr,\r)\vdash \leta\,z:(1,\r)}
		\end{mathpar}
		In the innermost application, we have an important constant $\leta$ and a productive variable $z$. The important constant has value 1. The value of $z$ is 0 because, even though it is productive, any important constants it produces will be just moved from the argument substituted for $z$. No important constant will be lost or produced during this process. There are no duplicates of variables in the application $\leta\,z$, so the value of the application is equal to sum of values, that is, 1.
		\begin{mathpar}
		\inferrule*[right=$(@)$]{
			y:(\pr,(\pr,\r)\arr\r)\vdash y:(0,(\pr,\r)\arr\r)\\
			z:(\pr,\r)\vdash \leta\,z:(1,\r)
		}{z:(\pr,\r),\ y:(\pr,(\pr,\r)\arr\r)\vdash y\,(\leta\,z):(1,\r)}
		\end{mathpar}
		We apply $y$ to $\leta\,z$. Variable $y$ is productive and takes a productive argument, which means that it incorporates its argument into the tree it produces and somehow increases the number of important constants in the process. However, this increase is computed in the lambda-term that is substituted for $y$, not here, hence it has value 0.
		\begin{mathpar}
		\inferrule*[right=$(@)$]{
			y:(\pr,(\pr,\r)\arr\r)\vdash y:(0,(\pr,\r)\arr\r)\\
			z:(\pr,\r),\ y:(\pr,(\pr,\r)\arr\r)\vdash y\,(\leta\,z):(1,\r)
		}{z:(\pr,\r),\ y:(\pr,(\pr,\r)\arr\r)\vdash y\,(y\,(\leta\,z)):(2,\r)}
		\end{mathpar}
		Both sides have $y$ in the environment here, so the duplication factor in $y\,(y\,(\leta\,z))$ is 1. We add it to the sum of values from both sides of the application, obtaining 2.
		\begin{mathpar}
		\inferrule*[right=$(\lambda)$]{
			\inferrule*[Right=$(\lambda)$]{
			z:(\pr,\r),\ y:(\pr,(\pr,\r)\arr\r)\vdash y\,(y\,(\leta\,z)):(2,\r)
		}{y:(\pr,(\pr,\r)\arr\r)\vdash \lambda z. y\,(y\,(\leta\,z)):(2, (\pr,\r) \arr \r)}
		}{\vdash \lambda y. \lambda z. y\,(y\,(\leta\,z)):(2,(\pr,(\pr,\r)\arr\r) \arr (\pr,\r) \arr \r)}
		\end{mathpar}
		The lambda-binders move $z$ and $y$ from the environment into the type without changing the value. The final value for this closed lambda-term is 2, as it always adds one important constant to the B\"ohm tree and, when something that increases the number of important constants is substituted for $y$, it applies that twice, causing at least one extra important constant to be added compared to what just one instance of $y$ would do.
	
		Note that it is possible to derive other type judgments for the lambda-term $y\,(y\,(\leta\,z))$. For example, the value would be equal to one if $y$ was not productive, even if $z$ was productive:
		\[
		y : (\np,(\pr,\r)\arr\r),\ z:(\pr,\r) \vdash y\,(y\,(\leta\,z)) : (1,\r).
		\]
		The rationale is that while $y$ takes productive arguments, it uses them exactly once and does not add an important constant to the B\"ohm tree either, so applying it any number of times will not change the number of important constants.
	\end{example}

\section{Soundness}\label{sec:soundness}

	In this section we prove the soundness of the type system, as described by following \lcnamecref{lem:soundness-fin},
	from which the $(2)\Rightarrow(1)$ implication of \cref{thm:main} follows immediately:

	\begin{lemma}\label{lem:soundness-fin}
		If we can derive $\emptyset\vdash M:(v,\r)$, where $M$ is a closed lambda-term of sort $\otyp$,
		then $\BT(M)$ has a finite branch containing at least $v$ occurrences of $\leta$.
	\end{lemma}

	We prove \cref{lem:soundness-fin} as follows.
	First, as in work of Parys~\cite{diagonal-types,parys-survey} we observe that instead of working directly with infinite lambda-terms $M$,
	we can ``cut off'' parts of $M$ not involved in the finite derivation of $\emptyset\vdash M:(v,\r)$, (i.e., subterms used as arguments for which no type pair is required).
	Formally, by cutting off we mean replacing by lambda-terms of the form $\lambda x_1\lamdots\lambda x_k.\upomega$,
	where the variables $x_1,\dots,x_k$ are chosen so that the sort of the lambda-term is appropriate.
	In consequence, it is enough to prove \cref{lem:soundness-fin} for finite lambda-terms.
	
	Second, we repeatedly use \cref{lem:soundness-red} below, reducing $M$ to its beta-normal form $N$, and never decreasing the productivity value.
	Finally, we observe that the beta-normal form~$N$ is simply a tree;
	a derivation of $\emptyset\vdash N:(v',\r)$, using only the rules for constants and application,
	describes some branch of this tree containing exactly $v'$ occurrences of the important constant~$\leta$.
	It remains to justify the following \lcnamecref{lem:soundness-red}, which describes a single beta-reduction:

	\begin{lemma}\label{lem:soundness-red}
		If we can derive $\emptyset\vdash M:(v,\r)$, where $M$ is a finite closed lambda-term of sort $\otyp$,
		and $M$ is not in the beta-normal form,
		then we can derive $\emptyset\vdash N:(v',\r)$ for a lambda-term $N$ such that $M\redb N$, and for some $v'$ satisfying $v\leq v'$.
	\end{lemma}

	While proving this \lcnamecref{lem:soundness-red}, we consider the leftmost outermost redex, $(\lambda x.K)\,L$.
	Thanks to this choice, $L$ is necessarily a closed subterm.
	This simplifies the situation: it is enough to consider empty type environments
	(we remark, however, that \cref{lem:soundness-red} can be shown in a similar way for every reduction, not only for the leftmost outermost reduction).
	We want to replace every subderivation $D$ for a type judgment $\emptyset\vdash(\lambda x.K)\,L:(w,\tau)$ concerning this redex
	with a derivation $D'$ for $\Gamma\vdash K[L/x]:(w',\tau)$.
	We obtain $D'$ by appropriately reorganizing subderivations of $D$:
	we take the subderivation of $D$ concerning $K$, we replace every leaf deriving a type $\sigma$ for $x$ by the subderivation of $D$ deriving this type $\sigma$ for $L$,
	and we update type environments and productivity values appropriately.

	Notice that every subderivation concerning $L$ is moved to at least one leaf concerning $x$ (nothing can disappear).
	The only reason why the value of the derivation can decrease is that potentially a productive type binding $x:(\pr,\sigma)$ was duplicated (say, $n$ times) in the derivation concerning $K$.
	In $D'$ this binding is no longer present (in $K[L/x]$ there is no $x$) so the value decreases by $n$,
	but in this situation the subderivation deriving $\sigma$ for $L$ becomes inserted in $n+1$ leaves.
	This subderivation is productive, so by creating $n$ additional copies of this subderivation we increase the value at least by $n$, compensating the loss caused by elimination of $x$.
	This implies that $w\leq w'$, and consequently $v\leq v'$.
	Check \cref{app:soundness} for more details.

\section{Completeness for safe lambda-terms}\label{sec:completeness}

	In this section we prove completeness for a subclass of lambda-terms called safe lambda-terms.
	The question whether completeness holds for general lambda-terms is left open.

	A lambda-term $M$ is \emph{superficially safe} when for every free variable $x$ of $M$ it holds that $\ord(M)\leq\ord(x)$.
	A lambda-term $M$ is \emph{safe} if it is superficially safe, and if for its every subterm of the form $K\,L_1\,\dots\,L_k$, where $K$ is not an application and $k\geq 1$,
	all subterms $K,L_1,\dots,L_k$ are superficially safe.
	This definition of safety coincides with the definitions from Salvati and Walukiewicz~\cite{schemes-lY}, and from Blum and Ong~\cite{blum-ong-safe}.

	Completeness for safe lambda-terms is given by the following \lcnamecref{lem:completeness-fin}.

	\begin{lemma}\label{lem:completeness-fin}
		For every $m\in\Nat$ there exists a function $H_m\colon\Nat\to\Nat$ such that if $M$ is a closed safe lambda-term of sort $\otyp$ and complexity at most $m$,
		and in $\BT(M)$ there is a finite branch having at least $n$ occurrences of $\leta$,
		then we can derive $\emptyset\vdash M:(v,\r)$ for some $v$ such that $n\leq H_m(v)$.
	\end{lemma}

	While proving this \lcnamecref{lem:completeness-fin}, it is convenient to split the productivity value into two parts.
	Given some fixed $\ell\in\Nat$, instead of type judgments of the form $\Gamma\vdash N:(v,\tau)$
	we consider extended type judgments of the form $\Gamma\vdash N:(u\oplus_\ell w,\tau)$, where $u+w=v$.
	On $u$ we accumulate only duplication factors concerning variables of order at least $\ell$, and on $w$ the remaining part of the value
	(i.e., duplication factors concerning variables of order smaller than $\ell$, plus the number of rules for the constant $\leta$).
	Having this definition, we can state a counterpart of \cref{lem:soundness-red}:

	\begin{lemma}\label{lem:compl-1red}
		Suppose that we can derive $\Gamma\vdash K[L/x]:(u\oplus_\ell w,\tau)$, where $L$ is closed, has order $\ell$, and does not use any variables of order at least $\ell$.
		Then we can derive $\Gamma\vdash (\lambda x.K)\,L:(u'\oplus_\ell w',\tau)$ for some $u',w'$ such that $2^u\cdot w\leq 2^{u'}\cdot w'$ and $u+w=0\Rightarrow u'+w'=0$.
	\end{lemma}

	Similarly to \cref{lem:soundness-red}, the derivation concerning $(\lambda x.K)\,L$ in the above \lcnamecref{lem:compl-1red}
	is obtained by appropriately reorganizing subderivations of the derivation concerning $K[L/x]$.
	In the type derivation concerning $K[L/x]$, there are some (zero or more) subderivations concerning $L$.
	A difficulty is caused by the fact that
	the same type pair $(f,\sigma)$ may be derived for multiple copies of $L$ in $K[L/x]$, using different subderivations.
	For every such $(f,\sigma)$ derived for $L$ we should choose just one subderivation, so that we can use it for the only copy of $L$ in $(\lambda x.K)\,L$.
	We choose the subderivation that provides the largest second component of the value; the other subderivations are removed.
	The value of the new derivation decreases, because some subderivations are removed, and increases, because we have a new variable $x$, that may cause some duplications.

	It remains to see that the value $u'\oplus_\ell w'$ of the new derivation satisfies the inequality $2^u\cdot w\leq 2^{u'}\cdot w'$.
	Consider some type pair $(f,\sigma)$ derived for $L$.
	If $f=\np$, the value of the removed subderivations is $0$, and the duplications of $x:(f,\sigma)$ in the type environment are not counted,
	so such a type pair does not cause any change of the value.
	Suppose that $f=\pr$, and that we have removed $n$ subderivations proving the type pair $(f,\sigma)$ for $L$;
	this may decrease the second component of the value at most $n+1$ times (the $n$ removed subderivations have values not greater than the one that remains).
	Simultaneously, in $K$ we use $n+1$ times the variable $x$ with this type pair $(f,\sigma)$, which increases the value (the total duplication factor) by $n$.
	The key point is that $L$ does not use any variables of order at least~$\ell$,
	and thus the first component (concerning variables of order at least $\ell$) of the value of subderivations for $L$ is $0$.
	Thus, we decrease the second component of the value at most $n+1$ times,
	and we increase the first component of the value by at least $n$.
	Since $2^n\geq n+1$, we obtain the required inequality between values.

	It is also important that $L$ is closed so that after removing some subderivations concerning $L$, the type environment of the whole derivation remains unchanged.
	This finishes the proof of \cref{lem:compl-1red} (a more formal proof, containing all details, can be found in \cref{app:compl-1red}).

	We now come back to the proof of \cref{lem:completeness-fin}.
	First, as in the previous section, we can assume that $M$ is finite by ``cutting off'' (i.e., replacing with $\lambda x_1\lamdots\lambda x_k.\upomega$)
	its parts not needed for producing the selected finite branch of $\BT(M)$.
	
	Second, it is convenient to assume here that $M$ is homogeneous.
	A sort $\alpha_1\arr\dots\arr\alpha_k\arr\otyp$ is \emph{homogeneous} if $\ord(\alpha_1)\geq\dots\geq\ord(\alpha_k)$ and all $\alpha_1,\dots,\alpha_k$ are homogeneous.
	A lambda-term is homogeneous if its every subterm has a homogeneous sort.
	It is known that every finite closed safe lambda-term $M$ of sort $\otyp$ can be converted into a lambda-term $M'$ that is additionally homogeneous~\cite{homogeneous},
	but has the same beta-normal form.
	Analyzing how $M$ is transformed into $M'$ (in~\cite{homogeneous}), it is tedious but straightforward to check that
	we can derive $\emptyset\vdash M:(v,\r)$ if and only if we can derive $\emptyset\vdash M':(v,\r)$;
	this is done in \cref{app:homogeneous}.

	It is thus enough to prove \cref{lem:completeness-fin} assuming that $M$ is finite and homogeneous.
	To this end, we consider a particular sequence of reductions leading from $M$ to its beta-normal form $\BT(M)$.
	Namely, whenever a lambda-term $N$ reached so far (i.e., after some number of reductions) from $M$ is of complexity $\ell+1$,
	then we reduce in $N$ a redex $(\lambda x.K)\,L$ such that $\ord(\lambda x.K)=\ell+1$,
	and no variables of order at least $\ell$ occur in $L$, both as free variables and as bound variables; we call such a redex \emph{$\ell$-good}.
	Such a redex always exists: among redexes with $\ord(\lambda x.K)=\ell+1$ it is enough to choose the rightmost one.
	The complexity of a lambda-term cannot increase during a beta-reduction,
	thus in the sequence of reductions we can find lambda-terms $M_m,M_{m-1},\dots,M_0$, where $M_m=M$, and $M_0=\BT(M)$, and the complexity of every $M_{\ell+1}$ is at most $\ell+1$,
	and $M_\ell$ can be reached from $M_{\ell+1}$ by a sequence of $\ell$-good reductions.

	Homogeneity is preserved during beta-reductions.
	Moreover, $\ell$-good beta-reductions preserve safety (while this is not true for arbitrary beta-reductions).
	Indeed, homogeneity for an $\ell$-good redex $(\lambda x.K)\,L$ implies that $\ord(L)=\ell$ (the first argument is of the highest order, i.e., of order $\ell$),
	thus safety implies that all free variables of $L$ are of order at least $\ell$.
	But, by the definition of an $\ell$-good redex, no variables of order at least $\ell$ occur in $L$.
	It follows that $L$ is closed.
	While substituting a closed lambda-term $L$ for $x$, all superficially safe subterms of $K$ remain superficially safe (no new free variables appear).
	Thus the lambda-term after the $\ell$-good beta-reduction remains safe.

	Recall the sequence of lambda-terms $M_m,M_{m-1},\dots,M_0$ defined above.
	Assuming that in $\BT(M)=M_0$ there is a finite branch having exactly $n$ occurrences of $a$, we can derive $\emptyset\vdash M_0:(n,\r)$ using only the rules for application and constants.
	Suppose now that we can derive $\emptyset\vdash M_\ell:(v_\ell,\r)$ for some $\ell\in\set{0,\dots,m-1}$.
	We want to find a derivation of $\emptyset\vdash M_{\ell+1}:(v_{\ell+1},\r)$, where for arbitrarily large values $v_\ell$, the values $v_{\ell+1}$ are also arbitrarily large
	(more precisely, we obtain the inequality $v_\ell\leq 2^{v_{\ell+1}}$).
	Notice that between $M_{\ell+1}$ and $M_\ell$ there may be arbitrarily many ($\ell$-good) reductions.
	Consider thus some $\ell$-good redex $(\lambda x.K)\,L$ that is reduced at some moment between $M_{\ell+1}$ and $M_\ell$.
	By definition, $L$ does not use any variables of order at least $\ell$ and, as already observed, $L$ is closed;
	thus, \cref{lem:compl-1red} can be applied.
	We start with $\emptyset\vdash M_\ell:(0\oplus_\ell v_\ell,\r)$ (notice that in $M_\ell$ there are no variables of order $\ell$ or higher, so the first component of the value is $0$).
	Then, we use \cref{lem:compl-1red} for every $\ell$-good beta-reduction between $M_{\ell+1}$ and $M_\ell$.
	This leads to $\emptyset\vdash M_{\ell+1}:(u_{\ell+1}\oplus_\ell w_{\ell+1},\r)$,
	where $v_\ell\leq 2^{u_{\ell+1}}\cdot w_{\ell+1}\leq 2^{v_{\ell+1}}$ (taking $v_{\ell+1}=u_{\ell+1}+w_{\ell+1}$), as required.

	The function $H_m$, appearing in the statement of \cref{lem:completeness-fin}, can be defined as a tower of powers of $2$ of height $m$: $H_0(v)=v$ and $H_{\ell+1}(v)=H_\ell(2^v)$.
	Then from $n\leq H_\ell(v_\ell)$ it follows that $n\leq H_{\ell+1}(v_{\ell+1})$; since $n=H_0(v_0)$, we thus obtain the desired inequality $n\leq H_m(v_m)$.

\section{The algorithm}\label{sec:algorithm}

	Both design and implementation of our algorithm is based on the \textsc{HorSat} algorithm,
	a saturation-based algorithm for model checking recursion schemes against alternating automata by Broadbent and Kobayashi~\cite{horsat}.

	Our type system was described in previous sections to work with any infinitary lambda-term.
	The algorithm, in turn, inputs infinitary lambda-terms represented in a finite way, in the form of a recursion scheme.
	To be concrete, we make this representation explicit now:
	A recursion scheme $\Gg=(\Sigma,\Nn,\Rr,X_\mathsf{st})$ consists of
	\begin{itemize}
	\item	a signature $\Sigma$ (i.e., a set of constants with assigned arities),
	\item	a set $\Nn$ of nonterminals with assigned sorts (formally, nonterminals are just distinguished variables),
	\item	a mapping $\Rr$ from nonterminals in $\Nn$ to finite lambda-terms such that $\Rr(X)$ is of the same sort as $X$,
		has no free variables other than nonterminals from $\Nn$, and is of the form $\lambda x_1\lamdots\lambda x_n.K$, where the subterm $K$ does not contain any lambda-binders, and
	\item	a starting nonterminal $X_\mathsf{st}\in\Nn$ of sort $\otyp$.
	\end{itemize}
	We assume that elements of $\Nn$ are not used as bound variables, and that $\Rr(X)$ is not a nonterminal.
	Furthermore, as in previous sections, we assume for simplicity that $\Sigma=\set{\leta,\letb,\letc}$, where the constants $\leta,\letb,\letc$ have arity $1,2,0$, respectively
	(the implementation, however, accepts arbitrary signatures).
	
	Given a recursion scheme $\Gg$, and a lambda-term $M$ (possibly containing some nonterminals from $\Nn$),
	let $\Lambda_\Gg(M)$ be the lambda-term obtained as a limit of applying repeatedly the following operation to $M$:
	take an occurrence of some nonterminal $X$, and replace it by $\Rr(X)$
	(the nonterminals should be chosen so that every nonterminal is eventually replaced).
	We remark that while substituting $\Rr(X)$ for a nonterminal $X$, there is no need for any renaming of variables (capture-avoiding substitution),
	since $\Rr(X)$ does not have free variables other than nonterminals.
	The infinitary lambda-term \emph{represented by} $\Gg$ is defined as $\Lambda_\Gg(X_\mathsf{st})$, and denoted $\Lambda(\Gg)$.
	Observe that $\Lambda(\Gg)$ is a closed lambda-term of sort $\otyp$.
	We say that $\Gg$ is \emph{safe} if $\Rr(X)$ is safe for every $X\in\Nn$; then $\Lambda(\Gg)$ is safe as well.

	The goal of the algorithm is to determine whether $\Lambda(\Gg)$ for a given recursion scheme $\Gg$ is unbounded or not.
	In the light of \cref{thm:main}, this boils down to checking whether one can derive $\emptyset\vdash\Lambda(\Gg):(v,\r)$ for arbitrarily large values of $v$.
	The idea is to find a single derivation with a productive ``loop''; by repeating this loop, one can increase the productivity value arbitrarily.
	We make this more precise now.

	First, we determine type pairs that can be derived for non-terminals.
	Namely, for $v\in\Nat$ let $\p(v)=\pr$ if $v>0$ and $\pr(v)=\np$ otherwise.
	Let $\Tt_\Gg$ be the smallest set containing all bindings $X:(\p(v),\tau)$ (with $X\in\Gg$) such that we can derive $\emptyset\vdash\Rr(X):(v,\tau)$,
	potentially using $\emptyset\vdash Y:(w,\sigma)$ for $(Y:(\p(w),\sigma))\in\Tt_\Gg$ as assumptions (i.e., as additional typing rules).
	
	Note that productivity values in derivations may be shifted (i.e., increased/decreased by a constant):
	we can derive $\emptyset\vdash\Rr(X):(v,\tau)$ out of an assumption $\emptyset\vdash Y:(w,\sigma)$
	if and only if we can derive $\emptyset\vdash\Rr(X):(v+(w'-w),\tau)$ out of an assumption $\emptyset\vdash Y:(w',\sigma)$ with $\p(w)=\p(w')$.
	It is thus enough to try assumptions $\emptyset\vdash Y:(w,\sigma)$ with $w=0$ and $w=1$ only.
	
	We have only finitely many bindings that may be potentially added to $\Tt_\Gg$,
 	as there are finitely many possible types per sort and each nonterminal of $\Gg$ has a fixed sort.
	Moreover, taking into account the above, there are only finitely many derivations that may be potentially created for lambda-terms $\Rr(X)$.
	It follows that $\Tt_\Gg$ may be computed using saturation (i.e., as the least fixed point).
	The following lemma is immediate:
	
	\begin{lemma}
		For any nonterminal $X\in\Nn$ and any type pair $(f,\tau)$
		one can derive $\Lambda_\Gg(X):(v,\tau)$ for some $v$ with $\p(v)=f$ if and only if $(X:(f,\tau))\in\Tt_\Gg$.
	\end{lemma}
	
	Next, knowing which type judgments may be derived, we want to detect a productive cycle.
	To this end, we create a graph with bindings from $\Tt_\Gg$ as nodes, called a \emph{derivation graph}.
	We draw an edge from $X:(\p(v),\tau)$ to $Y:(\p(w),\sigma)$ if one can derive $\emptyset\vdash\Rr(X):(v,\tau)$ using $\emptyset\vdash Y:(w,\sigma)$ as an assumption,
	and potentially using some other assumptions $\emptyset\vdash Z:(u,\rho)$ with $(Z:(\p(u),\rho))\in\Tt_\Gg$
	(the assumption $\emptyset\vdash Y:(w,\sigma)$ necessarily has to be used, at least once).
	Such an edge is called \emph{productive} if $v>w$.
	Note that this may happen for three reasons: 1) in the derivation there is some positive duplication factor or an important constant,
	2) some other assumption $\emptyset\vdash Z:(u,\rho)$ with $u>0$ is used, or 3) the assumption $\emptyset\vdash Y:(w,\sigma)$ is used more than once and $w>0$. 
	Note that edges of the derivation graphs can be incrementally computed at the time of computing $\Tt_\Gg$.
	We now obtain the main theorem, adding Point (3) equivalent to Point (2) from \cref{thm:main}:

	\begin{theorem}\label{thm:main2}
		The following two statements are equivalent for every recursion scheme $\Gg=(\Sigma,\Nn,\Rr,X_\mathsf{st})$:
		\begin{enumerate}[(1)]
		\setcounter{enumi}{1}
		\item	one can derive $\emptyset\vdash\Lambda(\Gg):(v,\r)$ for arbitrarily large values of $v$,
		\item	$(X_\mathsf{st}:(\pr,\r))\in\Tt_\Gg$ and the derivation graph contains a cycle with a productive edge, reachable from $X_\mathsf{st}:(\pr,\r)$.
		\end{enumerate}
	\end{theorem}
	
	By \cref{thm:main}, assuming additionally that $\Gg$ is safe, Point (2) (and hence Point (3) as well) is also equivalent to the property we want to check:
	\begin{enumerate}[(1)]
		\item	$\Lambda(\Gg)$ is unbounded.
	\end{enumerate}
	
	\begin{proof}
		Point (2) follows from Point (3) quite directly.
		Indeed, an edge from $X:(\pr,\tau)$ to $Y:(\pr,\sigma)$ means that any derivation of $\emptyset\vdash\Lambda_\Gg(Y):(w,\sigma)$
		can be extended to a derivation of $\emptyset\vdash\Lambda_\Gg(X):(v,\tau)$ for some $v\geq w$
		(in the sense that the latter contains the former as a subderivation),
		where $v>w$ if the edge is productive.
		Following edges on the cycle with a productive edge we can thus extend a derivation of $\emptyset\vdash\Lambda_\Gg(X):(v,\tau)$
		into a larger derivation of $\emptyset\vdash\Lambda_\Gg(X):(v',\tau)$, where $v'>v$;
		doing this repeatedly increases the productivity value arbitrarily.
		Finally, we can use the path from $X_\mathsf{st}:(\pr,\r)$ to the cycle, making the created derivation a part of a derivation concerning $\Lambda(\Gg)=\Lambda_\Gg(X_\mathsf{st})$.
		
		Let us move on to the implication from Point (2) to Point (3).
		First, consider all possible derivations of $\emptyset\vdash\Rr(X):(v,\tau)$ (with arbitrary $X\in\Nn$, $v$, $\tau$)
		that may use assumptions (additional typing rules) $\emptyset\vdash Y:(v',\tau')$ for nonterminals $Y$, as in the definition of $\Tt_\Gg$.
		Because $\Rr(X)$ for every $X\in\Nn$ is a finite lambda-term, and $|\Nn|$ is finite as well,
		there exists a bound $B$ such that any derivation as above
		uses assumptions for nonterminals in at most $B$ leaves,
		and simultaneously the growth of productivity value caused by duplication factor or important constants in the derivation is at most $B$.
	
		Next, we prove (by induction on the size of a derivation) that if we can derive $\emptyset\vdash\Lambda_\Gg(X):(v,\tau)$ with $v\geq(2B)^k$ for some $k\in\Nat$,
		then $(X:(\p(v),\tau))\in\Tt_\Gg$ and in the derivation graph there is a path from this node having at least $k$ productive edges.
		Indeed, out of a derivation of $\emptyset\vdash\Lambda_\Gg(X):(v,\tau)$ we can reconstruct a derivation of $\emptyset\vdash\Rr(X):(v,\tau)$;
		whenever the former derivation contains a (strictly smaller) subderivation concerning $\Lambda_\Gg(Y)$ for some nonterminal $Y$,
		in the latter we use an assumption concerning $Y$, which is in $\Tt_\Gg$ by the induction hypothesis.
		This already shows that $(X:(\p(v),\tau))\in\Tt_\Gg$.
		If $k=0$, we are done.
		Suppose that $k\geq 1$.
		By properties of the bound $B$, in the derivation of $\emptyset\vdash\Rr(X):(v,\tau)$ there has to be an assumption $\emptyset\vdash Y:(v',\tau')$
		with $v'\geq\frac{v-B}{B}\geq\frac{v}{2B}$
		(from $v$ we subtract $B$ for duplication factors and important constants inside the derivation of $\emptyset\vdash\Rr(X):(v,\tau)$,
		and we divide the remaining part of the productivity value into at most $B$ assumptions),
		that is, we can derive $\emptyset\vdash\Lambda_\Gg(Y):(v',\tau')$.
		If $v'=v\geq(2B)^k$, then we have an edge to a node from which there is a path having at least $k$ productive edges, by the induction hypothesis.
		Otherwise $v>v'\geq(2B)^{k-1}$, so we have a productive edge to a node, from 
		which there is a path having at least $k-1$ productive edges, by the induction hypothesis.
		
		Starting with $X=X_\mathsf{st}$ and $k>|\Tt_\Gg|$ we obtain a path from $(X_\mathsf{st}:(\pr,\r))$ containing more than $|\Tt_\Gg|$ productive edges;
		this path has to reach a cycle with a productive edge, as needed for Point (3).
	\end{proof}

\subsection{Implementation}
 
Our implementation is based on the implementation of \textsc{HorSat}~\cite{horsat}
(more precisely: of \textsc{HorSat2}, a revised version of the \textsc{HorSat} algorithm, due to Kobayashi and Terao).
The main reason behind the efficiency of our (and \textsc{HorSat}'s) implementation is that we do not compute all possible bindings in $\Tt_\Gg$.
We derive only types which may be useful in a derivation of a type of the starting nonterminal $X_\mathsf{st}$.

To this end, we first perform a 0-CFA analysis of the input to compute which lambda-terms may flow into particular nonterminal parameters.
For example, if we have a nonterminal $X$ with a parameter $x$, we find all applications in the form $M_X\,N$, where $M_X$ is a lambda-term where
the head is $X$ or may eventually be substituted by $X$; then the lambda-term $N$ is flowing into the parameter $x$. Note that 0-CFA analysis gives an overestimation.
For example, for a nonterminal $Y$ with two parameters $x$ and $y$ it is possible to write two applications, $Y\,M\,N$ and $Y\,M'\,N'$,
transformed in such a way that, according to 0-CFA output, $M$ may flow into $x$ at the same time as $N'$ into $y$.
This behavior can exponentially (with respect to the number of arguments) increase the number of possible typings of $Y$ that the algorithm has to check.

Thanks to 0-CFA analysis, we have a complete list of lambda-terms which may be substituted for parameters of nonterminals, which limits the set of computed types of nonterminals while still retaining all types required to type $X_\mathsf{st}$.
This is sufficient to start a loop where new types are found.
In this loop, nonterminals (more precisely: the lambda-terms $\Rr(X)$ for nonterminals $X$) are typed using information about types flowing into their parameters.
New types of nonterminals enable finding new types of lambda-terms in which these nonterminals occur.
These lambda-terms again flow into parameters of some nonterminals, so the new types enable finding additional typings of these nonterminals, returning to the beginning of the loop.
This loop continues until a fixpoint, that is, until no new typings of nonterminals or lambda-terms flowing into their parameters can be computed.
As there is a finite number of possible typings of nonterminals and lambda-terms in any given recursion scheme, this fixpoint will be reached after finitely many steps.
During this loop, each time we type a nonterminal, we take note of typings of nonterminals that were used in the derivation (i.e., are subtrees in the final derivation tree),
and incrementally construct the derivation graph described earlier in this section.
This way we find all parts of the derivation graph that are reachable from $X_{\mathsf{st}} \colon (\pr, \r)$; the optimizations made by 0-CFA only remove unreachable parts.

We perform typing of terms without lambda-bindings in a top-down manner, that is, we iterate over all possible types an application may have, find compatible types for its left-hand side,
and then type-check its right-hand side under each possible environment with its desired, already known type.
The top-down approach is particularly efficient when the constant $\letb$ is present in the term, since its argument with type $\top$ does not have to be type-checked. This is also the case for nonterminals with arguments of type $\top$.

We also include three
optimizations similar to the ones present in \textsc{HorSat}:
\begin{itemize}
    \item After a new type of a nonterminal is found, we type all lambda-terms that contain it in a way where the new typing is used at least once.
    	This optimization can be turned off with a flag \texttt{-nofntty}.
    \item After a new type of a lambda-term flowing into a parameter $x$ of
        a nonterminal $X$ is found,
    	we search for new typings of $\Rr(X)$ (or its subterms that flow into parameters of other nonterminals) in a way where at least one instance of the parameter $x$ has the new type.
    	This optimization can be turned off with a flag \texttt{-noftty}.
    \item When no parameter of a nonterminal $X$ is used as a left-hand side of an application in $\Rr(X)$,
    	we infer types of parameters of $X$ from left-hand sides of applications instead of using
    	types of lambda-terms flowing into these parameters.
    	Then $\Rr(X)$ does not have to be typed whenever a new type is computed for lambda-terms flowing into its parameters.
    	This optimization can be turned off with a flag \texttt{-nohvo}.
\end{itemize}

\subsection{Benchmarks}

\begin{figure}[th]
    \centering
    \resizebox{\textwidth}{!}{
    \begin{tabular}{|l|l|l|l|l|l|l|l|l|l|l|}
    \hline
    & & \multicolumn{8}{c|}{Flags and run times in seconds} \\
    \hline
    Benchmark name &
    \begin{sideways}$\ord(\Gg)$\end{sideways} &
    \begin{sideways}no flags\end{sideways} &
    \begin{sideways}\texttt{-noftty}\end{sideways} &
    \begin{sideways}\texttt{-nofntty}\end{sideways} &
    \begin{sideways}\texttt{-nohvo}\end{sideways} &
    \begin{sideways}\parbox{1.6cm}{\texttt{-noftty -nofntty}}\end{sideways} &
    \begin{sideways}\parbox{1.6cm}{\texttt{-noftty -nohvo}}\end{sideways} &
    \begin{sideways}\parbox{1.6cm}{\texttt{-nofntty -nohvo}}\end{sideways} &
    \begin{sideways}\parbox{1.6cm}{\texttt{-noftty -nofntty -nohvo}}\end{sideways} \\
    \hline
    \texttt{jwig-cal\_main}               & 2 & 0.942     & 0.930     & 1.572     & 0.611     & 1.574            & 0.607          & 0.728           & 0.783               \\
    \texttt{spec\_cps\_coerce1-c}          & 3 & 66.28     & 3.921     & 3.285     & 102.2     & 0.647            & 3.517          & 38.105          & 4.557               \\
    \texttt{xhtmlf-div-2}                & 2 & 9.591     & 9.555     & 7.457     & 9.238     & 7.426            & 9.007          & 7.387           & 7.478               \\
    \texttt{xhtmlf-m-church}             & 2 & 9.689     & 9.655     & 7.481     & 9.235     & 7.470            & 9.172          & 7.459           & 7.562               \\
    \hline
    \texttt{fold\_fun\_list}               & 7 & 0.425     & 0.025     & OOM       & 0.402     & OOM              & 0.024          & OOM             & OOM                 \\
    \texttt{fold\_right}                  & 5 & 25.01     & 0.017     & 0.277     & 23.75     & 0.028            & 0.016          & 0.265           & 0.028               \\
    \texttt{search-e-church}             & 6 & 127.7     & 14.96     & 338.1     & 119.5     & 12.19            & 14.11          & 325.9           & 13.53               \\
    \texttt{zip}                         & 4 & 2.618     & 153.4     & 64.79     & 27.30     & 201.5            & 150.0          & 68.63           & 98.81               \\
    \hline
    \texttt{filepath}                    & 2 & TO        & OOM       & OOM       & OOM       & OOM              & OOM            & OOM             & OOM                 \\
    \texttt{filter-nonzero-1}            & 5 & 59.69     & 1.635     & 37.56     & 64.19     & 2.636            & 1.641          & 39.52           & 1.922               \\
    \texttt{filter-nonzero}              & 5 & 0.641     & 0.106     & 0.570     & 0.670     & 0.138            & 0.106          & 0.653           & 0.083               \\
    \texttt{map-plusone-2}               & 5 & 5.384     & 1.082     & 2.591     & 5.065     & 0.851            & 1.092          & 5.633           & 1.125               \\
    \hline
    \texttt{cfa-life2}                    & 14 & TO       & TO        & OOM       & TO        & TO               & TO             & TO              & TO                  \\
    \texttt{cfa-matrix-1}                 & 8 & 1.794    & 2.376     & 1.672     & 1.708     & 2.124            & 2.310          & 1.582           & 2.280               \\
    \texttt{cfa-psdes}                    & 7 & 0.017    & 0.022     & 0.019     & 0.026     & 0.022            & 0.018          & 0.041           & 0.029               \\
    \texttt{tak\_inf}                      & 8 & 10.48    & 5.842     & 12.16     & 9.759     & 11.88            & 5.739          & 11.58           & 8.858               \\
    \hline
    \texttt{dna}                          & 2 & 0.118    & 0.106     & 0.050     & 0.115     & 0.074            & 0.109          & 0.055           & 0.074               \\
    \texttt{fibstring}                    & 4 & 0.022    & 0.024     & 0.016     & 0.023     & 0.020            & 0.023          & 0.015           & 0.020               \\
    \texttt{g45}         & 4 & 41.12    & 45.73     & 7.819     & 16.46     & 56.61            & 46.54          & 1.629           & 46.06               \\
    \texttt{l}                            & 3 & 0.012    & 0.023     & 0.016     & 0.017     & 0.017            & 0.022          & 0.015           & 0.017               \\
    \hline
    \texttt{fib02\_odd\_fin}               & 4 & 0.007     & 0.014     & 0.009     & 0.013     & 0.013            & 0.013          & 0.014           & 0.013               \\
    \texttt{fib\_even\_inf}                & 4 & 0.545     & 2.854     & 0.697     & 0.505     & 2.555            & 2.637          & 0.697           & 2.604               \\
    \texttt{two\_add\_inf}                 & 4 & 0.022     & 0.018     & 0.023     & 0.024     & 0.015            & 0.012          & 0.021           & 0.015               \\
    \texttt{two\_succ\_inf}                & 4 & 0.005     & 0.006     & 0.005     & 0.006     & 0.004            & 0.005          & 0.006           & 0.005               \\
    \hline
    \end{tabular}
    }
    TO means timeout (10 minutes), OOM means out of memory error.
    \caption{\textsc{InfSat} benchmark results in all combinations of optimization flags.}
    \label{fig:benchmarks}
\end{figure}

We prepared benchmarks for the implementation of our algorithm by modifying benchmarks presented in the \textsc{HorSat} paper \cite{horsat}.
\textsc{HorSat} is an algorithm that efficiently checks whether an alternating tree automaton (ATA) accepts the B\"ohm tree of the lambda-term defined by a recursion scheme. This problem is different from ours, however, it also performs an analysis on the same trees and is also $m$-\textsf{EXPTIME}-complete, where the difficulty depends on the order $m$ of the lambda-term. 
Hence, we use benchmark results presented in the \textsc{HorSat} paper \cite{horsat} as an indication that the terms used there are difficult to analyze.

The details on the original benchmarks and their origin can be found in the \textsc{HorSat} paper \cite{horsat}. According to authors of the aforementioned paper, these benchmarks contain practical data. Indeed, some of them model analysis of an XHTML document or a short computer program. Analysis of many of them was not completed in a reasonable time by other model checking algorithms similar to \textsc{HorSat}, while \textsc{HorSat} only failed to analyze benchmark \texttt{fibstring} in a reasonable time.

Let us describe how we modified the benchmarks to fit our analysis. We left the original recursion scheme intact and selected a few important constants that are present close to leaves of generated trees to increase the difficulty. The result was a decision problem whether the operations described by important constants in modelled programs could be executed unbounded number of times, assuming the program halted. Additionally, we added four benchmarks specific to \textsc{InfSat} that answer mathematical questions such as whether there exist arbitrarily large odd numbers defined as Church numerals.

We present results of our benchmarks in Figure \ref{fig:benchmarks}. They were performed on a laptop with Intel Core i5-7600K, 16GB RAM.
We consider these results a success, as only two benchmarks on practical data failed to compute within ten minutes. As \textsc{InfSat} is the first efficient algorithm solving the simultaneous unboundedness problem, we do not have any data to which we could compare our benchmark results. At this stage, we can assess effectiveness of our optimizations. We can see that using all optimization flags produced good results consistently, however, each optimization happened to slow down some benchmarks.
The reason is that optimizations turned off by \texttt{-noftty} and \texttt{-nofntty} add substantial polynomial overhead when generating list of possible environments
and optimization turned off by \texttt{-nohvo} changes the way types of some terms are computed.
However, all of them can exponentially reduce the computation time in many cases which is why they are turned on by default.

\section{Multi-letter case}\label{sec:multi-letter}

	It is not difficult to modify the type system to handle simultaneous unboundedness.
	In this part, instead of a single important constant $\leta$ of arity $1$, we consider important constants $\leta_1,\dots,\leta_s$, all of arity $1$ ($\setP$ is the set containing them all).
	Besides them, in the signature $\Sigma$ we have a constant $\letb$ of arity $2$, and constants $\letc$ and $\upomega$ of arity $0$.

	In the type system instead of a single productivity flag in $\set{\pr,\np}$,
	we have a productivity set, being a subset of $\setP$, and saying which important letters are produced.
	Likewise, instead of a productivity value in $\Nat$, we have a productivity function $v\colon\setP\to\Nat$, specifying a value separately for every letter.
	The rules of the type system are adopted in the expected way:
	$v(\leta_i)$ increases when we use constant $\leta_i$, and when we duplicate a type binding for a type pair $(A,\sigma)$ with $\leta_i\in A$.

	One more change is, however, necessary.
	Indeed, while proving \cref{lem:compl-1red} in the multi-letter case, we cannot choose a single subderivation concerning $L$ that has the greatest value,
	since now a value is not a number, but rather a function.
	Instead, for every constant $\leta_i$ we can choose a subderivation concerning $L$ for which the $i$-th coordinate of the value is the greatest.
	We thus need to allow $s$ (i.e., $|\setP|$) subderivations for every type pair.

	To this end, on the left of the arrow in a type, we do not have a set of types, but rather a multiset of types, where we allow to have at most $s$ copies of every type.
	Likewise, in the type environment we have at most $s$ copies of every type binding.

	Formally, an \emph{$s$-multiset} is a multiset that contains at most $s$ copies of every element.
	The set of $s$-multisets of elements of $X$ is denoted $\Pp_{\leq s}(X)$.
	Notice that $1$-multisets are just sets, hence $\Pp_{\leq 1}(X)=\Pp(X)$.
	A union of $s$-multisets $U,V$, denoted $U\cup V$, is defined as follows:
	if $U$ and $V$ contain, respectively, $n$ and $m$ copies of an element $x$, then $U\cup V$ contains $\min(n+m,s)$ copies of this element.
	We use the notation $\setof{x_i}{i\in I}$ for $\bigcup_{i\in I}\set{x_i}$, where $\set{x_i}$ is the multiset containing $x_i$ once.

	For each sort $\alpha$ we define the set $\Tt^\alpha_s$ of types of sort $\alpha$ as follows:
	\begin{align*}
		\Tt^\otyp_s=\set{\r},&&
		\Tt^{\alpha\arr\beta}_s=\Pp_{\leq s}(\Pp(\setP)\times\Tt^\alpha_s)\times\Tt^\beta_s.
	\end{align*}
	We again use the notation with $\arr$ and $\bigwedge$, but this time while writing $\bigwedge_{i\in I}(f_i,\tau_i)\arr\tau$,
	we assume that any pair occurs as $(f_i,\tau_i)$ at most $s$ times.

	A \emph{type judgment} is of the form $\Gamma\vdash M:(v,\tau)$, where $v\colon\setP\to\Nat$, and where we require that the type $\tau$ and the lambda-term $M$ are of the same sort.
	The \emph{type environment} $\Gamma$ is an $s$-multiset of bindings of variables of the form $x^\alpha:(A,\tau)$,
	where $A\subseteq\setP$ is a \emph{productivity set} and $\tau\in\Tt^\alpha_s$ is a type.
	For $a\in\setP$ by $\Gamma\restr_a$ we denote the $s$-multiset of those binding from $\Gamma$ that have $a$ in their productivity set.

	By $\zero$ we denote the function from $\setP$ to $\Nat$ that maps every $\leta_i$ to $0$,
	and by $\chi_i$ the function that maps $\leta_i$ to $1$ and all other $\leta_j$ to $0$.
	The type system consists of the following rules:
	\begin{mathpar}
	\inferrule{}{
		\emptyset\vdash \leta_i:(\chi_i,(A,\r)\arr\r)
	}\and
	\inferrule{}{
		\emptyset\vdash \letc:(\zero,\r)
	}\\
	\inferrule{}{
		\emptyset\vdash \letb:(\zero,(A,\r)\arr\top\arr\r)
	}\and
	\inferrule{}{
		\emptyset\vdash \letb:(\zero,\top\arr(A,\r)\arr\r)
	}\and
	\inferrule{}{
		x:(A,\tau)\vdash x:(\zero,\tau)
	}
	\and
	\inferrule*[right=($\lambda$)]{
		\Gamma\cup\setof{x:(A_i,\tau_i)}{i\in I}\vdash K:(v,\tau)
	\\
		x\not\in \dom(\Gamma)
	}{
		\Gamma\vdash\lambda x.K:(v,\textstyle\bigwedge_{i\in I}(A_i,\tau_i)\arr\tau)
	}
	\end{mathpar}

	We define the duplication factor $\dupl((\Gamma_j)_{j\in J})$, as the function (from $\setP$ to $\Nat$) that maps every important constant $a\in\setP$ to
	$\sum_{j\in J}\left|\Gamma_j\restr_a\right|-\left|\bigcup_{j\in J}\Gamma_j\restr_a\right|$.
	\begin{mathpar}
	\inferrule*[right=$(@)$]{
		0\not\in I
	\\
		\forall i\in I.\ A_i=\setof{a\in\setP}{v_i(a)>0 \lor \Gamma_i\restr_a\neq\emptyset}
	\\
		\Gamma_0\vdash K:(v_0,\textstyle\bigwedge_{i\in I}(A_i,\tau_i)\arr\tau)
	\\
		\Gamma_i\vdash L:(v_i,\tau_i)\mbox{ for each }i\in I
	}{
		\textstyle\bigcup_{i\in\set{0}\cup I}\Gamma_i\vdash K\,L:(\dupl((\Gamma_i)_{i\in\set{0}\cup I})+\textstyle\sum_{i\in\set{0}\cup I}v_i,\tau)
	}
	\end{mathpar}
	Recall that this time any pair may occur as $(f_i,\tau_i)$ at most $s$ times.

\bibliographystyle{plainurl}
\bibliography{bib}

\newpage
\appendix

\section{Soundness}\label{app:soundness}

	Below we complete details of the proofs from \cref{sec:soundness}.
	The proofs are purely syntactic, and follow the sketch given in \cref{sec:soundness}; nothing really interesting happens here.

	We conduct the proofs in the multi-letter case (which generalizes the single-letter case considered in \cref{sec:soundness}).
	Let us first state a \lcnamecref{lem:soundness-sub} useful while proving \cref{lem:soundness-red}, describing a substitution (recall that $s$ is the number of important letters in the alphabet):

	\begin{lemma}\label{lem:soundness-sub}
		If we can derive
		\begin{align*}
			&\Gamma\cup\setof{x:(A_i,\tau_i)}{i\in I}\vdash M:(v_0,\tau)&&\mbox{and}\\
			&\emptyset\vdash N:(v_i,\tau_i)&&\mbox{for all $i\in I$,}
		\end{align*}
		where any pair occurs as $(A_i,\tau_i)$ for at most $s$ indices $i\in I$, and $0\not\in I$, and $x\not\in\dom(\Gamma)$, and $A_i=\setof{a\in\setP}{v_i(a)>0}$ for all $i\in I$,
		then we can derive $\Gamma\vdash M[N/x]:(v',\tau)$ for some $v'$ such that for all $a\in\setP$ we have that
		\begin{align*}
			&\sum_{i\in\set{0}\cup I}v_i(a)\leq v'(a)&&\mbox{and}\\
			&\sum_{i\in\set{0}\cup I}v_i(a)=0\ \Rightarrow\ v'(a)=0.
		\end{align*}
	\end{lemma}
	
	\begin{proof}
		We prove the \lcnamecref{lem:soundness-sub} by induction on the size of a chosen derivation of $\Gamma\cup\setof{x:(A_i,\tau_i)}{i\in I}\vdash M:(v_0,\tau)$.
		Assume the statement holds for all smaller derivations, and consider the last rule in the type derivation concerning $M$.
		There are four cases.
		
		First, $M$ may be either of the form $\lambda x.K$ (with the variable $x$ for which we substitute), or a constant, or a variable other than $x$.
		Then necessarily $I=\emptyset$ and $M[N/x]=M$, so thesis is obtained trivially by simply taking the input derivation concerning~$M$.

		Second, we may have $M=x$.
		Then necessarily $\Gamma$ is empty, $I$ is a singleton, $v_0=\zero$, and $M[N/x]=N$.
		This time the thesis is obtained trivially by taking the (only) input derivation concerning $N$.

		Third, if $M$ is of the form $\lambda y.K$ for $y\neq x$, then the last rule applied in the derivation concerning $M$ is
		\begin{mathpar}
		\inferrule*{
			\Gamma\cup\setof{x:(A_i,\tau_i)}{i\in I}\cup\setof{y:(B_j,\sigma_j)}{j\in J}\vdash K:(v_0,\sigma)
		}{
			\Gamma\cup\setof{x:(A_i,\tau_i)}{i\in I}\vdash\lambda y.K:(v_0,\textstyle\bigwedge_{j\in J}(B_j,\sigma_j)\arr\sigma)
		}
		\end{mathpar}
		and $y\not\in\dom(\Gamma)$.
		By the induction hypothesis, we can derive
		\begin{align*}
			\Gamma\cup\setof{y:(B_j,\sigma_j)}{j\in J}\vdash K[N/x]:(v',\sigma)
		\end{align*}
		for some $v'$ such that $\sum_{i\in\set{0}\cup I}v_i(a)\leq v'(a)$ and $\sum_{i\in\set{0}\cup I}v_i(a)=0\Rightarrow v'(a)=0$ for all $a\in\setP$.
		By applying the $(\lambda)$ rule we derive
		\begin{align*}
			\Gamma\vdash(\lambda y.K)[N/x]:(v',\textstyle\bigwedge_{j\in J}(B_j,\sigma_j)\arr\sigma),
		\end{align*}
		proving the statement.

		Finally, assume that $M$ is in the form $K\,L$, so the last rule in the derivation concerning $M$ is
		\begin{mathpar}
		\inferrule*{
			\Gamma_0'\vdash K:(w_0,\textstyle\bigwedge_{j\in J}(B_j,\sigma_j)\arr\tau)
		\\
			\Gamma_j'\vdash L:(w_j,\sigma_j)\mbox{ for each }j\in J
		}{
			\Gamma\cup\setof{x:(A_i,\tau_i)}{i\in I}\vdash K\,L:(\dupl((\Gamma_j')_{j\in\set{0}\cup J})+\textstyle\sum_{j\in\set{0}\cup J}w_j,\tau)
		}
		\end{mathpar}
		where $0\not\in J$.
		The type environments $\Gamma_j'$ for $j\in\set{0}\cup J$ can be represented as
		\begin{align}\label{eq:gamma-union}
			\Gamma_j'=\Gamma_j\cup\setof{x:(A_i,\tau_i)}{i\in I_j},
		\end{align}
		where $x\not\in\dom(\Gamma_j)$ and $I_j$ is a subset of $I$;
		then $\Gamma=\bigcup_{j\in\set{0}\cup J}\Gamma_j$ and
		\begin{align}\label{eq:setsI}
			\setof{x:(A_i,\tau_i)}{i\in I}=\bigcup_{j\in\set{0}\cup J}\setof{x:(A_i,\tau_i)}{i\in I_j}
		\end{align}
		(where these are unions of $s$-multisets).
		Moreover, in the above representation we should choose the subsets $I_j$ of $I$ in a ``balanced way'', that is,
		we should ensure that $\bigcup_{j\in\set{0}\cup J} I_j=I$ (this is a union of standard sets, not $s$-multisets).
		This is possible because, by \cref{eq:setsI}, if the number of occurrences of a variable binding in $\setof{x:(A_i,\tau_i)}{i\in I}$ is $k$,
		then its number of occurrences altogether in all $\setof{x:(A_i,\tau_i)}{i\in I_j}$ for $j\in J$ is at least $k$.
		By assumptions of the $(@)$ rule we also have that
		\begin{align}
			B_j&=\setof{a\in\setP}{w_j(a)>0\lor\Gamma_j'\restr_a\neq\emptyset}\nonumber\\
			&=\setof{a\in\setP}{w_j(a)>0\lor\Gamma_j\restr_a\neq\emptyset\lor\exists i\in I_j.\ v_i(a)>0}\label{eq:1}
		\end{align}
		for every $j\in J$ (the second equality holds because $A_i=\setof{a\in\setP}{v_i(a)>0}$).
		By the induction hypothesis we can derive
		\begin{align*}
			\Gamma_0\vdash K[N/x]:\left(w_0',\textstyle\bigwedge_{j\in J}(B_j,\sigma_j)\arr\tau\right)
		\end{align*}
		and, for each $j\in J$,
		\begin{align*}
			\Gamma_j\vdash L[N/x]:(w_j',\sigma_j),
		\end{align*}
		where
		\begin{align}
			&w_j(a)+\sum_{i\in I_j}v_i(a)\leq w_j'(a)&&\mbox{and}\label[inequality]{eq:2}\\
			&w_j(a)+\sum_{i\in I_j}v_i(a)=0\ \Rightarrow\ w_j'(a)=0\label[implication]{eq:3}
		\end{align}
		for all $j\in\set{0}\cup J$ and $a\in\setP$.
		Note that \cref{eq:3} can be actually changed into an equivalence: the right-to-left implication follows from \cref{eq:2}.
		This equivalence, together with \cref{eq:1}, implies that
		\begin{align*}
			B_j&=\setof{a\in\setP}{w_j'(a)>0\lor\Gamma_j\restr_a\neq\emptyset}.
		\end{align*}
		This allows us to apply the $(@)$ rule and derive the desired type judgment
		\begin{align*}
			\textstyle\bigcup_{j\in\set{0}\cup J}\Gamma_j
			\vdash
			(K\,L)[N/x]:(\dupl((\Gamma_j)_{j\in\set{0}\cup J})+\textstyle\sum_{j\in\set{0}\cup J}w_j',\tau).
		\end{align*}
		
		It remains to prove required conditions on the productivity values.
		Fix some $a\in\setP$, and recall that $v'$ and $v_0$ are the productivity values in, respectively, the resulting and the original type judgments, that is,
		\begin{align}
			v'(a)&=\dupl((\Gamma_j)_{j\in\set{0}\cup J})(a)+\sum_{j\in\set{0}\cup J}w_j'(a),&&\mbox{and}\label{eq:def-vp}\\
			v_0(a)&=\dupl((\Gamma_j')_{j\in\set{0}\cup J})(a)+\sum_{j\in\set{0}\cup J}w_j(a).\label{eq:def-v0}
		\end{align}
		Applying \cref{eq:2} to \cref{eq:def-vp} we can write
		\begin{align}\label[inequality]{eq:vp}
			\dupl((\Gamma_j)_{j\in\set{0}\cup J})(a)+\sum_{j\in\set{0}\cup J}w_j(a)+\sum_{j\in\set{0}\cup J}\sum_{i\in I_j}v_i(a)\leq v'(a).
		\end{align}
		Simultaneously \cref{eq:gamma-union,eq:def-v0} imply that
		\begin{align}
			\sum_{i\in\set{0}\cup I}v_i(a)&=\dupl((\Gamma_j)_{j\in\set{0}\cup J})(a)+\dupl((\setof{x:(A_i,\tau_i)}{i\in I_j})_{j\in\set{0}\cup J})(a)\nonumber\\
				&\phantom{={}}+\sum_{j\in\set{0}\cup J}w_j(a)+\sum_{i\in I}v_i(a).\label{eq:vi}
		\end{align}
		Thus, taking into account \cref{eq:vp,eq:vi}, the inequality $\sum_{i\in\set{0}\cup I}v_i(a)\leq v'(a)$, which we aim to prove, reduces to
		\begin{align}\label[inequality]{eq:goal}
			\dupl((\setof{x:(A_i,\tau_i)}{i\in I_j})_{j\in\set{0}\cup J})(a)\leq\sum_{j\in\set{0}\cup J}\sum_{i\in I_j}v_i(a)-\sum_{i\in I}v_i(a);
		\end{align}
		this is what we aim to prove.
		Let us define $[a\in A_i]$ to be $1$ if $a\in A_i$ and $0$ otherwise.
		With this notation, expanding the definition of $\dupl$ and using \cref{eq:setsI} we can write
		\begin{align*}
			\dupl((\setof{x:(A_i,\tau_i)}{i\in I_j})_{j\in\set{0}\cup J})(a)=\sum_{j\in\set{0}\cup J}\sum_{i\in I_j}[a\in A_i]-\sum_{i\in I}[a\in A_i].
		\end{align*}
		We substitute this to \cref{eq:goal}.
		Observe that, for any fixed $i\in I$, both the component $[a\in A_i]$ on the left side and the component $v_i(a)$ on the right side
		is added exactly $|\setof{j\in\set{0}\cup J}{i\in I_j}|-1$ times.
		We can thus reformulate our goal, \cref{eq:goal}, to
		\begin{align*}
			\sum_{i\in I}[a\in A_i]\cdot(|\setof{j\in\set{0}\cup J}{i\in I_j}|-1)
				\leq\sum_{i\in I}v_i(a)\cdot(|\setof{j\in\set{0}\cup J}{i\in I_j}|-1).
		\end{align*}
		Because $I=\bigcup_{j\in\set{0}\cup J}I_j$, for every $i\in I$ the number $|\setof{j\in\set{0}\cup J}{i\in I_j}|-1$ is nonnegative.
		simultaneously $a\in A_i\Leftrightarrow v_i(a)>0$, that is, $[a\in A_i]\leq v_i(a)$, by assumptions of the \lcnamecref{lem:soundness-sub}.
		This shows the above inequality.

		We also need to prove that $\sum_{i\in\set{0}\cup I}v_i(a)=0$ implies $v'(a)=0$.
		Assume thus the former.
		By \cref{eq:vi} we then have that $\dupl((\Gamma_j)_{j\in\set{0}\cup J})(a)=0$ and $w_j(a)+\sum_{i\in I_j}v_i(a)=0$ for all $j\in\set{0}\cup J$.
		The latter by \cref{eq:3} implies that $w_j'(a)=0$ for all $j\in\set{0}\cup J$,
		which by \cref{eq:def-vp} gives us the required thesis $v'(a)=0$.
		This finishes the proof for the case $M=K\,L$, and this the proof of the whole \lcnamecref{lem:soundness-sub}.
	\end{proof}

	We now use \cref{lem:soundness-sub} to show the following \lcnamecref{lem:soundness-red-multi}, which generalizes \cref{lem:soundness-red} to the multi-letter case:
	
	\begin{lemma}\label{lem:soundness-red-multi}
		If we can derive $\emptyset\vdash M:(v,\r)$, where $M$ is a finite closed lambda-term of sort $\otyp$,
		and $M$ is not in the beta-normal form,
		then we can derive $\emptyset\vdash N:(v',\r)$ for a lambda-term $N$ such that $M\redb N$,
		and for some $v'$ satisfying $v(a)\leq v'(a)$ and $v(a)=0\Rightarrow v'(a)=0$ for all $a\in\setP$.
	\end{lemma}

	\begin{proof}
		The proof is by induction on the size of $M$.
		Since $M$ is finite, closed, of sort $\otyp$, and not in the beta-normal form, we have the following possibilities for the shape of $M$:
		
		First, $M$ may be of the form $\leta_j\,K$ for some $j\in\set{1,\dots,s}$, where $K$ is again finite, closed, of sort $\otyp$, and not in the beta-normal form.
		A derivation of $\emptyset\vdash M:(v,\r)$ necessarily ends with the $(@)$ rule, above which, ``on the left'',	we have a type judgment concerning $\leta_j$.
		According to the rule for $\leta_j$, this type judgment is of the form $\emptyset\vdash\leta_j:(\chi_j,(A,\r)\arr\r)$.
		It follows that above the $(@)$ rule we also have $\emptyset\vdash K:(w,\r)$, where $v=w+\chi_j$, and where $A=\setof{a\in\setP}{w(a)\geq 0}$.
		Applying the induction hypothesis to $K$ we obtain a lambda-term $L$ and a function $w'$ such that $K\redb L$, and we can derive $\emptyset\vdash L:(w',\r)$,
		and $w(a)\leq w'(a)$ and $w(a)=0\Rightarrow w'(a)=0$ for all $a\in\setP$.
		Note that $A=\setof{a\in\setP}{w'(a)\geq 0}$.
		Thus, taking $N=\leta_j\,L$ and $v'=w'+\chi_j$ we can derive $\emptyset\vdash N:(v',\r)$, using the $(@)$ rule and the type judgment $\emptyset\vdash\leta_j:(\chi_j,(A,\r)\arr\r)$.
		We also have $v(a)\leq v'(a)$ and $v(a)=0\Rightarrow v'(a)=0$ for all $a\in\setP$.
		
		Second, $M$ may be of the form $\letb\,K_1\,K_2$, where $K_1,K_2$ are finite, closed, and of sort $\otyp$,
		and at least one of them, say $K_i$, is not in the beta-normal form.
		Here the derivation ends with two $(@)$ rules, above which we have either $\letb:(\zero,(A,\r)\arr\top\arr\r)$ or $\letb:(\zero,\top\arr(A,\r)\arr\r)$.
		This implies that above one of the $(@)$ rules we have no additional type judgments, and over the other we have have $\emptyset\vdash K_j:(v,\r)$ for some $j\in\set{1,2}$
		(note that the productivity value is unchanged).
		If $i=j$, we apply the induction hypothesis for the type judgment $\emptyset\vdash K_i:(v,\r)$;
		we obtain $L_i$ and $v'$ such that $K_i\redb L_i$, and $v(a)\leq v'(a)$ and $v(a)=0\Rightarrow v'(a)=0$ for all $a\in\setP$, and we can derive $\emptyset\vdash L_i:(v',\r)$.
		We then take $N=\letb\,L_1\,K_2$ if $i=1$ or $N=\letb\,K_1\,L_2$ if $i=2$, and we finish the derivation as before, obtaining $\emptyset\vdash N:(v',\r)$.
		If, contrarily, $i\neq j$ (i.e., the subterm not in the beta-normal form is not used in the derivation),
		we choose an arbitrary lambda-term $L_i$ such that $K_i\redb L_i$, we define $N$ as above, and we take $v'=v$.
		We can now derive $\emptyset\vdash N:(v',\r)$ by simply replacing $K_i$ with $L_i$.
		
		The only remaining situation is that $M$ is of the form $(\lambda x.K)\,L\,L_1\,\dots\,L_k$, where $L$ is closed.
		A derivation of $\emptyset\vdash M:(v,\r)$ ends with $k$ applications of the $(@)$ rule,
		above which we have a type judgment $\emptyset\vdash(\lambda x.K)\,L:(w,\tau)$ concerning $(\lambda x.K)\,L$, and some type judgments concerning the arguments $L_1,\dots,L_k$.
		The last two rules used while deriving the former type judgment are the $(@)$ rule and the $(\lambda)$ rule:
		\begin{mathpar}
			\inferrule*[leftskip=-0.6em]{
				\inferrule*[leftskip=0.6em]{
					\emptyset\cup\setof{x:(A_i,\tau_i)}{i\in I}\vdash K:(v_0,\tau)
				}{
					\emptyset\vdash\lambda x.K:(v_0,\textstyle\bigwedge_{i\in I}(A_i, \tau_i)\arr\tau)
				}
			\\
				\emptyset\vdash L:(v_i,\tau_i)\mbox{ for each } i\in I
			}{
				\emptyset\vdash(\lambda x.K)\,L:(w,\tau)
			}
		\end{mathpar}
		where any pair occurs as $(A_i,\tau_i)$ for at most $s$ indices $i\in I$, and $0\not\in I$, and $A_i=\setof{a\in\setP}{v_i(a)>0}$ for all $i\in I$,
		and $w=\sum_{i\in\set{0}\cup I}v_i$.
		These conditions are exactly as needed by \cref{lem:soundness-sub},
		which allows us to derive $\emptyset\vdash K[L/x]:(w',\tau)$ for some $w'$ such that $w(a)\leq w'(a)$ and $w(a)=0\Rightarrow w'(a)=0$.
		We then take $N=K[L/x]\,L_1\,\dots\,L_k$ and $v'=v+(w'-w)$.
		We derive $\emptyset\vdash N:(v',\r)$ by applying to $\emptyset\vdash K[L/x]:(w',\tau)$ the same $k$ rules as in the original derivation
		(only the productivity value becomes shifted by $w'-w$).
		We also have $v(a)\leq v'(a)$ and $v(a)=0\Rightarrow v'(a)=0$ for all $a\in\setP$, as needed.
	\end{proof}

	We now write a \lcnamecref{lem:sound-stop} describing lambda-terms in the beta-normal form:

	\begin{lemma}\label{lem:sound-stop}
		If we can derive $\emptyset\vdash M:(v,\r)$, where $M$ is a finite closed lambda-term of sort $\otyp$ in the beta-normal form,
		then $\BT(M)$ has a finite branch that for every $a\in\setP$ contains exactly $v(a)$ occurrences of~$a$.
	\end{lemma}

	\begin{proof}
		We prove the \lcnamecref{lem:sound-stop} by induction on the size of $M$.
		Since $M$ is finite, closed, of sort $\otyp$, and in the beta-normal form, the head of $M$ is necessarily a constant, that is,
		$M$ is of the form either $\leta_j\,M'$, or $\letb\,M_1\,M_2$, or $\letc$, or $\upomega$:
		\begin{itemize}
		\item	Suppose first that $M=\leta_j\,M'$.
			As already observed while proving the previous \lcnamecref{lem:soundness-red-multi}, a derivation of $\emptyset\vdash M:(v,\r)$ necessarily ends as follows:
			\begin{mathpar}
				\inferrule*{
					\emptyset\vdash\leta_j:(\chi_j,(A,\r)\arr\r)
				\and
					\emptyset\vdash M':(v-\chi_j,\r)
				}{
					\emptyset\vdash M:(v,\r)
				}
			\end{mathpar}
			By the induction hypothesis, $\BT(M')$ has a finite branch that for every $a\in\setP$ contains exactly $(v-\chi_j)(a)$ occurrences of~$a$.
			This branch extended by the initial $\leta_j$ gives us a branch of $\BT(M)$ as required
			(notice that $\BT(M)=\leta_j\,\BT(M')$, and that the missing $\chi_j$ is compensated by the constant $\leta_j$ from the root).
		\item	The case of $M=\letb\,M_1\,M_2$ is similar.
			The derivation may end in this way:
			\begin{mathpar}
				\inferrule*[leftskip=-5em,rightskip=-5.4em]{
					\inferrule*[leftskip=5em,rightskip=5.4em]{
						\emptyset\vdash\letb:(\zero,(A,\r)\arr\top\arr\r)
					\and
						\emptyset\vdash M_1:(v,\r)
					}{
						\emptyset\vdash\letb\,M_1:(\zero,\top\arr\r)
					}
				}{
					\emptyset\vdash M:(v,\r)
				}
			\end{mathpar}
			or in this way:
			\begin{mathpar}
				\inferrule*[leftskip=-0.3em]{
					\inferrule*[leftskip=0.3em]{
						\emptyset\vdash\letb:(\zero,\top\arr(A,\r)\arr\r)
					}{
						\emptyset\vdash\letb\,M_1:(\zero,(A,\r)\arr\r)
					}
					\and
						\emptyset\vdash M_2:(v,\r)
				}{
					\emptyset\vdash M:(v,\r)
				}
			\end{mathpar}
			Thus we have a derivation of either $\emptyset\vdash M_1:(v,\r)$ or $\emptyset\vdash M_2:(v,\r)$.
			We use the induction hypothesis for this type judgment, and we prepend $\letb$ to the resulting branch.
		\item	For $M=\letc$, we could derive $\emptyset\vdash M:(v,\r)$ only for $v=\zero$, so the only branch of $\BT(M)=\letc$ satisfies the thesis.
		\item	It is impossible to derive any type judgment concerning $\upomega$, so the case $M=\upomega$ could not hold.
		\qedhere\end{itemize}
	\end{proof}

	Finally, we prove a multi-letter counterpart of \cref{lem:soundness-fin} (the original \cref{lem:soundness-fin} can be obtained by simply taking $s=1$):

	\begin{lemma}\label{lem:soundness-fin-multi}
		If we can derive $\emptyset\vdash M:(v,\r)$, where $M$ is a closed lambda-term of sort $\otyp$,
		then $\BT(M)$ has a finite branch that for every $a\in\setP$ contains at least $v(a)$ occurrences of~$a$.
	\end{lemma}

	\begin{proof}
		Suppose first that $M$ is finite.
		In this case, we conduct the proof by induction on the maximal number $n$ such that there exists a sequence of $n$ beta-reductions from $M$
		(note that this number is finite thanks to the strong normalization property of simply-typed lambda-calculus).
		The base case is when $n=0$: no beta-reduction can be performed from $M$ (i.e., $M$ is in the beta-normal form); then the thesis follows directly from \cref{lem:sound-stop}.		
		Next, we have a case when $M$ is not in the beta-normal form.
		Then using \cref{lem:soundness-red} we obtain a derivation of $\emptyset\vdash N:(v',\r)$
		for some $v'$ satisfying $v\leq v'$ and some $N$ such that $M\redb N$ (which implies $\BT(M)=\BT(N)$).
		Clearly $N$ is again finite, closed, and of sort $\otyp$, and the maximal length of a sequence of beta-reductions starting in $N$ is smaller than for $M$.
		We can thus apply the induction hypothesis to $N$, and obtain a branch of $\BT(N)$ (i.e., of $\BT(M)$) as required.

		Finally, suppose that $M$ is infinite.
		A derivation of $\emptyset\vdash M:(v,\r)$ is a finite object that analyzes only a finite part of $M$, thus we can focus only on this finite part.
		To this end, we define a \emph{cut of $M$} to be
		a lambda-term obtained from $M$ by replacing some of its subterms with lambda-terms of the form $\lambda x_1\lamdots\lambda x_k.\upomega$,
		where $x_1,\ldots,x_k$ are chosen in a way that the sort of the subterm does not change.
		It is easy to see that if we can derive $\emptyset\vdash M:(v,\r)$, then there is some finite cut $M'$ of $M$ for which we can derive $\emptyset\vdash M':(v,\r)$:
		to obtain $M'$ we simply cut off subterms that are not involved in the derivation of $\emptyset\vdash M:(v,\r)$.
		On the other hand, every branch of $\BT(M')$ is also a branch of $\BT(M)$ (recall that, by definition of a branch, we do not consider branches ending with $\upomega$).
		Thus, the case of an infinite term $M$ reduces to the case of a finite term $M'$, which is already proven.
	\end{proof}

\section{Completeness}

	In this section we complete details of the proofs from \cref{sec:completeness}, conductive all proofs in the more general multi-letter case.
	
	Recall that in \cref{sec:completeness} we have introduced extended type judgments of the form $\Gamma\vdash N:(u\oplus_\ell w,\tau)$ for $\ell\in\Nat$,
	where on $u$ we accumulate only duplication factors concerning variables of order at least $\ell$, and on $w$ we accumulate the rest of the productivity value.
	This variant of type judgments can be used without any doubt also for the multi-letter case, when $u$ and $w$ are functions from $\setP$ to $\Nat$.
	
	In order to derive such extended type judgments we adjust rules of the type system in the expected way.
	Namely, for a type environment $\Gamma$ we define $\Gamma\restr_{\geq\ell}$ and $\Gamma\restr_{<\ell}$ to contain only those type bindings from $\Gamma$
	that concern variables of order at least $\ell$ and at most $\ell-1$, respectively.
	Then, rules of the type system can we written as follows:
	\begin{mathpar}
	\inferrule{}{
		\emptyset\vdash \letb:(\zero\oplus_\ell\zero,(A,\r)\arr\top\arr\r)
	}\and
	\inferrule{}{
		\emptyset\vdash \letb:(\zero\oplus_\ell\zero,\top\arr(A,\r)\arr\r)
	}\and
	\inferrule{}{
		\emptyset\vdash \leta_i:(\zero\oplus_\ell\chi_i,(A,\r)\arr\r)
	}\and
	\inferrule{}{
		\emptyset\vdash \letc:(\zero\oplus_\ell\zero,\r)
	}\and
	\inferrule{}{
		x:(A,\tau)\vdash x:(\zero\oplus_\ell\zero,\tau)
	}
	\and
	\inferrule*[right=($\lambda$)]{
		\Gamma\cup\setof{x:(A_i,\tau_i)}{i\in I}\vdash K:(u\oplus_\ell w,\tau)
	\\
		x\not\in \dom(\Gamma)
	}{
		\Gamma\vdash\lambda x.K:(u\oplus_\ell w,\textstyle\bigwedge_{i\in I}(A_i,\tau_i)\arr\tau)
	}
	\and
	\inferrule*[right=$(@)$]{
		0\not\in I
	\\
		\forall i\in I.\ A_i=\setof{a\in\setP}{u_i(a)+w_i(a)>0 \lor \Gamma_i\restr_a\neq\emptyset}
	\\
		u=\dupl((\Gamma_i\restr_{\geq\ell})_{i\in\set{0}\cup I})+\textstyle\sum_{i\in\set{0}\cup I}u_i
	\\
		w=\dupl((\Gamma_i\restr_{<\ell})_{i\in\set{0}\cup I})+\textstyle\sum_{i\in\set{0}\cup I}w_i
	\\
		\Gamma_0\vdash K:(u_0\oplus_\ell w_0,\textstyle\bigwedge_{i\in I}(A_i,\tau_i)\arr\tau)
	\\
		\Gamma_i\vdash L:(u_i\oplus_\ell w_i,\tau_i)\mbox{ for each }i\in I
	}{
		\textstyle\bigcup_{i\in\set{0}\cup I}\Gamma_i\vdash K\,L:(u\oplus_\ell w,\tau)
	}
	\end{mathpar}

	There exists a direct correspondence between the original type system and the new one:

	\begin{lemma}
		\begin{alphaenumerate}\label{lem:std-eqv-ext}
		\item If we can derive $\Gamma\vdash M:(u\oplus_\ell w,\tau)$, then we can also derive $\Gamma\vdash M:(u+w,\tau)$.
		\item If we can derive $\Gamma\vdash M:(v,\tau)$, where $M$ is of complexity at most $\ell$ and $\Gamma\restr_{\geq\ell}=\emptyset$,
			then we can also derive $\Gamma\vdash M:(\zero\oplus_\ell v,\tau)$.
		\end{alphaenumerate}
	\end{lemma}
	
	\begin{proof}
		Immediate induction on the size of a derivation.
		For the second item we have to observe that if $M=\lambda x.K$ then $x$ is of order at most $\ell-1$,
		implying that the type environment $\Gamma$ extended by some bindings for $x$ still has no bindings for variables of order at least $\ell$.
	\end{proof}

\subsection{Proof of Lemma~\ref{lem:compl-1red}}\label{app:compl-1red}

	As for the soundness proof, we need the following results that deals with a single substitution, which is then used while proving \cref{lem:compl-1red-multi}
	(a multi-letter counterpart of \cref{lem:compl-1red}):

	\begin{lemma}\label{lem:compl-sub}
		Suppose we can derive $\Gamma\vdash M[N/x]:(u\oplus_\ell w,\tau)$, where $N$ is closed, has order $\ell$, and does not use any variables of order at least $\ell$.
		In such a situation there exist a set $I$ of indices with $0\not\in I$, functions $u_0'$ and $w_i'$ for $i\in\set{0}\cup I$, and types $\tau_i$ for $i\in I$
		such that, taking $A_i=\setof{a\in\setP}{w_i'(a)>0}$ for $i\in I$,
		any pair occurs as $(A_i,\tau_i)$ for at most $s$ indices $i\in I$, and we can derive
		\begin{align*}
			&\Gamma\cup\setof{x:(A_i,\tau_i)}{i\in I}\vdash M:(u_0'\oplus_\ell w_0',\tau)&&\mbox{and}\\
			&\emptyset\vdash N:(\zero\oplus_\ell w_i',\tau_i)&&\mbox{ for all $i\in I$,}
		\end{align*}
		and for all $a\in\setP$ we have that
		\begin{align*}
			&2^{u(a)}\cdot w(a)\leq 2^{u'_0(a)}\cdot\sum_{i\in\set{0}\cup I}w_i'(a),\\
			&u(a)\leq u'_0(a),&&\mbox{and}\\
			&u(a)+w(a)=0\ \Rightarrow\ u'_0(a)+\sum_{i\in\set{0}\cup I}w_i'(a)=0.
		\end{align*}
	\end{lemma}

	\begin{proof}
		We prove the lemma by induction on the size of a chosen derivation of $\Gamma\vdash M[N/x]:(u\oplus_\ell w,\tau)$.
		Assume the statement holds for all smaller derivations, and consider the last rule in the input type derivation.
		There are four cases.
		
		First, $M$ may be either of the form $\lambda x.K$ (with the variable $x$ for which we substitute), or a constant, or a variable other than $x$.
		Then we have $M[N/x]=M$, and we take $I=\emptyset$, and $u_0'=u$, and $w'_0=w$.
		The input derivation proves $\Gamma\cup\setof{x:(A_i,\tau_i)}{i\in I}\vdash M:(u_0'\oplus_\ell w_0',\tau)$.
		Finally, $u(a)=u'_0(a)$ and $w(a)=\sum_{i\in\set{0}\cup I}w_i'(a)$ for all $a\in\setP$, proving the desired conditions on the productivity value.
		
		Second, we may have $M=x$.
		Analysing rules of the type system, we can easily show by induction on the size of a derivation that:
		\begin{itemize}
		\item	any derived type judgment concerning some lambda-term has in the type environment only type bindings for variables that are free in this lambda-term;
		\item	any derived type judgment concerning a lambda-term without any variables of order at least $\ell$ has $\zero$ in its ``$u$'' component of the productivity value.
		\end{itemize}
		Because by assumption $M[N/x]=N$ is closed and uses no variables of order at least $\ell$, this implies that $\Gamma=\emptyset$ and $u=\zero$.
		We take $I=\set{1}$, and $u_0'=w_0'=\zero$, and $w_1'=w$, and $\tau_1=\tau$; then $A_1=\setof{a\in\setP}{w_1'(a)>0}$.
		The input derivation proves $\emptyset\vdash N:(\zero\oplus_\ell w_1',\tau_1)$,
		and we can derive $\Gamma\cup\setof{x:(A_i,\tau_i)}{i\in I}\vdash M:(u_0'\oplus_\ell w_0',\tau)$, that is,
		$\set{x:(A_1,\tau)}\vdash x:(\zero\oplus_\ell\zero,\tau)$ using the rule for a variable.
		Finally, $u(a)=0=u'_0(a)$ and $w(a)=\sum_{i\in\set{0}\cup I}w_i'(a)$ for all $a\in\setP$, proving the desired conditions on the productivity value.

		Third, if $M$ is of the form $\lambda y.K$ for $y\neq x$, then the last rule applied in the input derivation is
		\begin{mathpar}
		\inferrule*{
			\Gamma\cup\setof{y:(B_j,\sigma_j)}{j\in J}\vdash K[N/x]:(u\oplus_\ell w,\sigma)
		}{
			\Gamma\vdash(\lambda y.K)[N/x]:(u\oplus_\ell w,\textstyle\bigwedge_{j\in J}(B_j,\sigma_j)\arr\sigma)
		}
		\end{mathpar}
		and $y\not\in\dom(\Gamma)$.
		By the induction hypothesis we obtain $I$, $u_0'$, $w_i'$ for $i\in\set{0}\cup I$, and $\tau_i$ for $i\in I$ such that,
		taking $A_i=\setof{a\in\setP}{w_i'(a)>0}$ for $i\in I$, we can derive
		\begin{align*}
			&\Gamma\cup\setof{x:(A_i,\tau_i)}{i\in I}\cup\setof{y:(B_j,\sigma_j)}{j\in J}\vdash K:(u_0'\oplus_\ell w_0',\sigma)&&\mbox{and}\\
			&\emptyset\vdash N:(\zero\oplus_\ell w_i',\tau_i)&&\mbox{for all $i\in I$,}
		\end{align*}
		and any pair occurs as $(A_i,\tau_i)$ for at most $s$ indices $i\in I$, and for all $a\in\setP$ we have that
		and $2^{u(a)}\cdot w(a)\leq 2^{u'_0(a)}\cdot\sum_{i\in\set{0}\cup I}w_i'(a)$, and $u(a)\leq u'_0(a)$, and
		$u(a)+w(a)=0\Rightarrow u'_0(a)+\sum_{i\in\set{0}\cup I}w_i'(a)=0$.
		By applying the $(\lambda)$ rule to the former type judgment we derive
		\begin{align*}
			\Gamma\cup\setof{x:(A_i,\tau_i)}{i\in I}\vdash\lambda y.K:(u_0'\oplus_\ell w_0',\textstyle\bigwedge_{j\in J}(B_j,\sigma_j)\arr\sigma),
		\end{align*}
		proving the statement.

		Finally, assume that $M$ is in the form $K\,L$, so the last rule in the input derivation is
		\begin{mathpar}
		\inferrule*{
			\Gamma_1\vdash K[N/x]:(u_1\oplus_\ell w_1,\textstyle\bigwedge_{j\in J}(B_j,\sigma_j)\arr\tau)
		\\
			\Gamma_j\vdash L[N/x]:(u_j\oplus_\ell w_j,\sigma_j)\mbox{ for each }j\in J
		}{
			\Gamma\vdash(K\,L)[N/x]:(u\oplus_\ell w,\tau)
		}
		\end{mathpar}
		where $1\not\in J$, where
		\begin{align}\label{eq:Bj}
			B_j=\setof{a\in\setP}{u_j(a)+w_j(a)>0\lor\Gamma_j\restr_a\neq\emptyset}
		\end{align}
		for $j\in J$, and where
		\begin{align}
			u&=\dupl((\Gamma_j\restr_{\geq\ell})_{j\in\set{1}\cup J})+\sum_{j\in\set{1}\cup J}u_j&&\mbox{and}&
			w&=\dupl((\Gamma_j\restr_{<\ell})_{j\in\set{1}\cup J})+\sum_{j\in\set{1}\cup J}w_j.\label[equalities]{eq:def-uw}
		\end{align}
		Potentially renaming the indices in $J$, we can also assume that $0\not\in J$.
		
		We apply the induction hypothesis for the type judgments above the final $(@)$ rule.
		We obtain sets of indices $I_j$ for $j\in\set{1}\cup J$;
		after potentially renaming indices in these sets, we can assume that these they are disjoint, and that each of them is disjoint from $\set{0,1}\cup J$.
		Going further, for each $j\in\set{1}\cup J$ we obtain a function $u_j'$, functions $w_i'$ for $i\in\set{j}\cup I_j$, and types $\tau_i$ for $i\in I_j$.
		Defining $I'=\bigcup_{j\in\set{1}\cup J}I_j$ and $A_i=\setof{a\in\setP}{w_i'(a)>0}$ for $i\in I'$ we have derivations for
		\begin{align*}
			&\Gamma_1\cup\setof{x:(A_i,\tau_i)}{i\in I_1}\vdash K:(u_1'\oplus_\ell w_1',\textstyle\bigwedge_{j\in J}(B_j,\sigma_j)\arr\tau),\\
			&\Gamma_j\cup\setof{x:(A_i,\tau_i)}{i\in I_j}\vdash L:(u_j'\oplus_\ell w_j',\sigma_j)&&\mbox{for each $j\in J$, and}\\
			&\emptyset\vdash N:(\zero\oplus_\ell w_i',\tau_i)&&\mbox{for each $i\in I'$.}
		\end{align*}
		Moreover, for each $j\in\set{1}\cup I$ we have that any pair occurs as $(A_i,\tau_i)$ for at most $s$ indices $i\in I_j$.
		We also have the conditions
		\begin{align}
			&2^{u_j(a)}\cdot w_j(a)\leq 2^{u_j'(a)}\cdot\sum_{i\in\set{j}\cup I_j}w_i'(a),\label[inequality]{eq:ujwj}\\
			&u_j(a)\leq u_j'(a),&&\mbox{and}\label[inequality]{eq:uj}\\
			&u_j(a)+w_j(a)=0\ \Rightarrow\ u'_j(a)+\sum_{i\in\set{j}\cup I_j}w_i'(a)=0\label[implication]{eq:impl}
		\end{align}
		for all $a\in\setP$ and $j\in\set{1}\cup J$.
		Note that \cref{eq:impl} can be actually changed into an equivalence: the right-to-left implication follows from \cref{eq:ujwj,eq:uj}.
		This equivalence, together with \cref{eq:Bj}, implies that
		\begin{align*}
			B_j&=\setof{a\in\setP}{u_j'(a)+u_j'(a)>0\lor\exists i\in I_j.\ w_i'(a)>0\lor\Gamma_j\restr_a\neq\emptyset}\\
			   &=\setof{a\in\setP}{u_j'(a)+u_j'(a)>0\lor(\Gamma_j\cup\setof{x:(A_i,\tau_i)}{i\in I_j})\restr_a\neq\emptyset}.
		\end{align*}
		
		We now come to the task of choosing the set of indices $I$ as an appropriate subset of $I'$, which is a nontrivial fragment of the proof.
		We remark that one cannot take all indices from $I'$ to $I$, because we have no guarantee that any pair occurs as $(A_i,\tau_i)$ for at most $s$ indices $i\in I'$,
		which is required in the statement of the lemma.
		Instead, for any pair $\gamma=(A,\rho)$ we denote the set $\setof{i\in I'}{(A_i,\tau_i)=\gamma}$ by $I'\restr_\gamma$.
		Then, for any pair $\gamma$ with $I'\restr_\gamma\neq\emptyset$ we proceed as follows:
		\begin{itemize}
		\item	first, for every $a\in\setP$ we take to $I$ this index from $I'\restr_\gamma$ for which $w_i'(a)$ is the largest (any such index in the case of a tie);
		\item	then, we ensure that exactly $\min(|I'\restr_\gamma|,s)$ indices from $I'\restr_\gamma$ are taken to $I$,
			by arbitrarily choosing some other indices from $I'\restr_\gamma$ and taking them to $I$.
		\end{itemize}
		Note that, in the first step, the same index $i\in I'\restr_\gamma$ may give the largest value of $w_i'(a)$ for multiple letters $a\in\setP$;
		then this index is taken to $I$ only once ($I$ is a set).
		It follows that, after this step, no more than $\min(|I'\restr_\gamma|,s)$ elements of $I'\restr_\gamma$ are taken to $I$,
		so the second step indeed makes sense.
		From the property that exactly $\min(|I'\restr_\gamma|,s)$ indices from $I'\restr_\gamma$ are taken to $I$, for every pair $\gamma$,
		we obtain that
		\begin{itemize}
		\item	any pair occurs as $(A_i,\tau_i)$ for at most $s$ indices $i\in I$, and
		\item	$\setof{x:(A_i,\tau_i)}{i\in I}=\bigcup_{j\in\set{1}\cup J}\setof{x:(A_i,\tau_i)}{i\in I_j}$.
		\end{itemize}
		Having this, we apply the $(@)$ rule to the type judgments obtained from the induction hypothesis; we derive
		\begin{align*}
			\Gamma\cup\setof{x:(A_i,\tau_i)}{i\in I}\vdash K\,L:(u_0'\oplus_\ell w_0',\tau),
		\end{align*}
		where (recall that $x$ is of order $\ell$)
		\begin{align}
			u_0'&=\dupl((\Gamma_j\restr_{\geq\ell}\cup\setof{x:(A_i,\tau_i)}{i\in I_j})_{j\in\set{1}\cup J})+\sum_{j\in\set{1}\cup J}u'_j\nonumber\\
			    &=\dupl((\Gamma_j\restr_{\geq\ell})_{j\in\set{1}\cup J})+\dupl((\setof{x:(A_i,\tau_i)}{i\in I_j})_{j\in\set{1}\cup J})+\sum_{j\in\set{1}\cup J}u'_j
				&&\mbox{and}\label{eq:def-u0p}\\
			w_0'&=\dupl((\Gamma_j\restr_{<\ell})_{j\in\set{1}\cup J})+\sum_{j\in\set{1}\cup J}w'_j.\label{eq:def-w0p}
		\end{align}
		
		It remains to prove the desired conditions on the productivity value.
		To this end, we fix some $a\in\setP$ and we consider the first inequality to be proven,
		\begin{align}\label[inequality]{eq:goal-c}
			2^{u(a)}\cdot w(a)\leq 2^{u'_0(a)}\cdot\sum_{i\in\set{0}\cup I}w_i'(a).
		\end{align}
		Using \cref{eq:def-uw} we expand the left side of \cref{eq:goal-c}:
		\begin{align*}
			2^{u(a)}\cdot w(a)&=2^{\dupl((\Gamma_j\restr_{\geq\ell})_{j\in\set{1}\cup J})(a)}\cdot 2^{\sum_{j\in\set{1}\cup J}u_j(a)}
				\cdot\dupl((\Gamma_j\restr_{<\ell})_{j\in\set{1}\cup J})(a)\\
			    &\phantom{={}}+2^{\dupl((\Gamma_j\restr_{\geq\ell})_{j\in\set{1}\cup J})(a)}\cdot 2^{\sum_{j\in\set{1}\cup J}u_j(a)}
				\cdot\sum_{j\in\set{1}\cup J}w_j(a).
		\end{align*}
		Likewise, using \cref{eq:def-u0p,eq:def-w0p} we expand the right side of \cref{eq:goal-c}:
		\begin{align*}
			2^{u'_0(a)}\cdot\sum_{i\in\set{0}\cup I}w_i'(a)\hspace{-8em}&\\
			&=2^{\dupl((\Gamma_j\restr_{\geq\ell})_{j\in\set{1}\cup J})(a)}
				\cdot 2^{\dupl((\setof{x:(A_i,\tau_i)}{i\in I_j})_{j\in\set{1}\cup J})(a)}\cdot 2^{\sum_{j\in\set{1}\cup J}u'_j(a)}\\
			    &\phantom{={}}\cdot\Big(\dupl((\Gamma_j\restr_{<\ell})_{j\in\set{1}\cup J})(a)+\sum_{j\in\set{1}\cup J}w'_j(a)+\sum_{i\in I}w_i'(a)\Big)\\
			&\geq 2^{\dupl((\Gamma_j\restr_{\geq\ell})_{j\in\set{1}\cup J})(a)}\cdot 2^{\sum_{j\in\set{1}\cup J}u'_j(a)}
				\cdot\dupl((\Gamma_j\restr_{<\ell})_{j\in\set{1}\cup J})(a)\\
			    &\phantom{={}}+ 2^{\dupl((\Gamma_j\restr_{\geq\ell})_{j\in\set{1}\cup J})(a)}\cdot 2^{\sum_{j\in\set{1}\cup J}u'_j(a)}
				\cdot\sum_{j\in\set{1}\cup J}w'_j(a)\\
			    &\phantom{={}}+2^{\dupl((\Gamma_j\restr_{\geq\ell})_{j\in\set{1}\cup J})(a)}
				\cdot 2^{\dupl((\setof{x:(A_i,\tau_i)}{i\in I_j})_{j\in\set{1}\cup J})(a)}\cdot 2^{\sum_{j\in\set{1}\cup J}u'_j(a)}
				\cdot\sum_{i\in I}w_i'(a).
		\end{align*}
		Note that the factor $2^{\dupl((\Gamma_j\restr_{\geq\ell})_{j\in\set{1}\cup J})(a)}$ occurs in front of every component, at both sides, hence we can ignore it.
		Using \cref{eq:uj} for every $j\in\set{1}\cup J$ we obtain that
		\begin{align*}
			2^{\sum_{j\in\set{1}\cup J}u_j(a)}\cdot\dupl((\Gamma_j\restr_{<\ell})_{j\in\set{1}\cup J})(a)\leq
				2^{\sum_{j\in\set{1}\cup J}u'_j(a)}\cdot\dupl((\Gamma_j\restr_{<\ell})_{j\in\set{1}\cup J}).
		\end{align*}
		Next, using \cref{eq:ujwj} for a single $j\in\set{1}\cup J$ and \cref{eq:uj} for all $j'\in\set{1}\cup J$ other than $j$,
		we obtain that
		\begin{align*}
			2^{\sum_{j'\in\set{1}\cup J}u_{j'}(a)}\cdot w_j(a)&\leq 2^{\sum_{j'\in\set{1}\cup J}u'_{j'}(a)}\cdot\sum_{i\in\set{j}\cup I_j}w_i'(a).
		\end{align*}
		We have the above inequality for every $j\in\set{1}\cup J$; by taking a sum and recalling that $I'=\biguplus_{j\in\set{1}\cup J}I_j$ we obtain that
		\begin{align*}
			2^{\sum_{j\in\set{1}\cup J}u_j(a)}\cdot\sum_{j\in\set{1}\cup J}w_j(a)
				&\leq 2^{\sum_{j\in\set{1}\cup J}u'_j(a)}\cdot\sum_{j\in\set{1}\cup J}w_j'(a)\\
				&\phantom{\leq{}}+2^{\sum_{j\in\set{1}\cup J}u'_j(a)}\cdot\sum_{i\in I'}w_i'(a).
		\end{align*}
		Thus, in order to prove \cref{eq:goal-c} it remains to show that
		\begin{align}\label[inequality]{eq:goal2}
			\sum_{i\in I'}w_i'(a)\leq2^{\dupl((\setof{x:(A_i,\tau_i)}{i\in I_j})_{j\in\set{1}\cup J})(a)}\cdot\sum_{i\in I}w_i'(a).
		\end{align}
		
		Recall that for every pair $\gamma$ we have defined the set $I'\restr_\gamma$; likewise we define $I\restr_\gamma=\setof{i\in I}{(A_i,\tau_i)=\gamma}$.
		We are going to prove that, for every pair $\gamma$,
		\begin{align}\label[inequality]{eq:goal2-gamma}
			\sum_{i\in I'\restr_\gamma}w_i'(a)\leq2^{\dupl((\setof{x:(A_i,\tau_i)}{i\in I_j})_{j\in\set{1}\cup J})(a)}\cdot\sum_{i\in I\restr_\gamma}w_i'(a).
		\end{align}
		Since $I=\biguplus_\gamma I\restr_\gamma$ and $I'=\biguplus_\gamma I'\restr_\gamma$, by summing \cref{eq:goal2-gamma} over all pairs $\gamma$ we obtain \cref{eq:goal2}.
		We have two cases.
		First, we may have a pair $\gamma=(A,\rho)$ with $a\not\in A$.
		Then, by definition of $A_i$ we have that $w_i'(a)=0$ for all $i\in I'\restr_\gamma$ and even more for all $i\in I\restr_\gamma\subseteq I'\restr_\gamma$;
		both sides of \cref{eq:goal2-gamma} are zero.
		Second, suppose that we have a pair $\gamma=(A,\rho)$ with $a\in A$.
		Then, the duplication factor $\dupl((\setof{x:(A_i,\tau_i)}{i\in I_j})_{j\in\set{1}\cup J})(a)$ includes, among others, $|I'\restr_\gamma|-|I\restr_\gamma|$,
		that is, $\dupl((\setof{x:(A_i,\tau_i)}{i\in I_j})_{j\in\set{1}\cup J})(a)\geq|I'\restr_\gamma|-|I\restr_\gamma|$.
		Let $k\in I'\restr_\gamma$ be the index (chosen to $I$) for which $w_i'(a)$ is the greatest.
		Then (using the inequality $n+1\leq 2^n$, being true for all $n\in\Nat$)
		\begin{align*}
			\sum_{i\in I'\restr_\gamma}w_i'(a)
			&\leq \sum_{i\in I\restr_\gamma\setminus\set{k}}w_i'(a)+(|I'\restr_\gamma|-|I\restr_\gamma|+1)\cdot w_k'(a)\\
			&\leq (|I'\restr_\gamma|-|I\restr_\gamma|+1)\cdot\sum_{i\in I\restr_\gamma}w_i'(a)
			\leq 2^{|I'\restr_\gamma|-|I\restr_\gamma|}\cdot\sum_{i\in I\restr_\gamma}w_i'(a)\\
			&\leq 2^{\dupl((\setof{x:(A_i,\tau_i)}{i\in I_j})_{j\in\set{1}\cup J})(a)}\cdot\sum_{i\in I\restr_\gamma}w_i'(a),
		\end{align*}
		as required.
		This finishes the proof of \cref{eq:goal-c}.
		
		The second inequality to be proven, $u(a)\leq u_0'(a)$,
		follows trivially from Equalities \labelcref{eq:def-uw} and~\labelcref{eq:def-u0p}, and from \cref{eq:uj} applied for all $j\in\set{1}\cup J$.
		
		Finally, we assume that $u(a)=w(a)=0$, and we prove that $u'_0(a)+\sum_{i\in\set{0}\cup I}w_i'(a)=0$.
		By \cref{eq:def-uw} the assumption implies that
		$u_j(a)=w_j(a)=0$ for all $j\in\set{1}\cup J$ and
		$\dupl((\Gamma_j\restr_{\geq\ell})_{j\in\set{1}\cup J})(a)=\dupl((\Gamma_j\restr_{<\ell})_{j\in\set{1}\cup J})(a)=0$.
		Using now \cref{eq:impl} we obtain that $u_j'(a)=w_j'(a)=0$ for all $j\in\set{1}\cup J$ and $w_i'(a)=0$ for all $i\in I'$ (so, in particular, for all $i\in I\subseteq I'$).
		Recalling the definition of the sets $A_i$ we see that $a\not\in A_i$ for all $i\in I'$,
		implying in particular that $\dupl((\setof{x:(A_i,\tau_i)}{i\in I_j})_{j\in\set{1}\cup J})(a)=0$.
		Using \cref{eq:def-u0p,eq:def-w0p} we obtain that $u_0'(a)=w_0'(a)=0$, giving us the thesis.
	\end{proof}

	We can now state and prove the following \lcnamecref{lem:compl-1red-multi}, being a multi-letter counterpart of \cref{lem:compl-1red}:

	\begin{lemma}\label{lem:compl-1red-multi}
		Suppose that we can derive $\Gamma\vdash K[L/x]:(u\oplus_\ell w,\tau)$, where $L$ is closed, has order $\ell$, and does not use any variables of order at least $\ell$.
		Then there exist $u',w'$ such that we can derive $\Gamma\vdash (\lambda x.K)\,L:(u'\oplus_\ell w',\tau)$ and for all $a\in\setP$ we have that
		\begin{align*}
			&2^{u(a)}\cdot w(a)\leq 2^{u'(a)}\cdot w'(a),\\
			&u(a)\leq u'(a),&&\mbox{and}\\
			&u(a)+w(a)=0\Rightarrow u'(a)+w'(a)=0.
		\end{align*}
	\end{lemma}

	\begin{proof}
		We start by applying \cref{lem:compl-sub} to the type judgment $\Gamma\vdash K[L/x]:(u\oplus_\ell w,\tau)$.
		We obtain $I$, $u_0'$, $w_i'$ for $i\in\set{0}\cup I$, and $\tau_i$ for $i\in I$
		such that, taking $A_i=\setof{a\in\setP}{w_i'(a)>0}$ for $i\in I$,
		any pair occurs as $(A_i,\tau_i)$ for at most $s$ indices $i\in I$, and we can derive
		\begin{align*}
			&\Gamma\cup\setof{x:(A_i,\tau_i)}{i\in I}\vdash K:(u_0'\oplus_\ell w_0',\tau)&&\mbox{and}\\
			&\emptyset\vdash L:(\zero\oplus_\ell w_i',\tau_i)&&\mbox{for all $i\in I$,}
		\end{align*}
		and for all $a\in\setP$ we have that
		\begin{align*}
			&2^{u(a)}\cdot w(a)\leq 2^{u'_0(a)}\cdot\sum_{i\in\set{0}\cup I}w_i'(a),\\
			&u(a)\leq u_0'(a),&&\mbox{and}\\
			&u(a)+w(a)=0\Rightarrow u'_0(a)+\sum_{i\in\set{0}\cup I}w_i'(a)=0.
		\end{align*}
		To the obtained type judgments we apply the $(\lambda)$ and $(@)$ rules, proving the thesis:
		\begin{mathpar}
			\hfill\inferrule*{
				\inferrule*{
					\Gamma\cup\setof{x:(A_i,\tau_i)}{i\in I}\vdash K:(u_0'\oplus_\ell w_0',\tau)
				}{
					\Gamma\vdash\lambda x.K:(u_0'\oplus_\ell w_0',\textstyle\bigwedge_{i\in I}(A_i,\tau_i)\arr\tau)
				}
			\and
				\emptyset\vdash L:(\zero\oplus w_i',\tau_i)\mbox{ for each }i\in I
			}{
				\Gamma\vdash(\lambda x.K)\, L:(u_0'\oplus_\ell(\textstyle\sum_{i\in\set{0}\cup I}w_i'),\tau)
			}\hfill\qedhere
		\end{mathpar}
	\end{proof}

	Next, we prove the same not only for the redex itself, but for any lambda-term containing such a redex:

	\begin{lemma}\label{lem:compl-1red-multi-term}
		Let $M,N$ be lambda-terms such that $N$ is obtained from $M$ by reducing a subterm of the form $(\lambda x.K)\,L$ (i.e., replacing one its occurrence with $K[L/x]$),
		where $L$ is closed, has order $\ell$, and does not use any variables of order at least $\ell$.
		If we can derive $\Gamma\vdash N:(u\oplus_\ell w,\tau)$, then there exist $u',w'$ such that we can derive $\Gamma\vdash M:(u'\oplus_\ell w',\tau)$ and for all $a\in\setP$ we have that
		\begin{align*}
			&2^{u(a)}\cdot w(a)\leq 2^{u'(a)}\cdot w'(a),\\
			&u(a)\leq u'(a),&&\mbox{and}\\
			&u(a)+w(a)=0\Rightarrow u'(a)+w'(a)=0.
		\end{align*}
	\end{lemma}
	
	\begin{proof}
		We prove the lemma by induction on the size of a chosen derivation of $\Gamma\vdash M:(u\oplus_\ell w,\tau)$.
		There are three possibilities.
		First, it may be the case that whole $M$ consists of the redex (i.e., $M=(\lambda x.K)\,L$ and $N=K[L/x]$);
		then we simply use \cref{lem:compl-1red-multi}.
		Second, it may be the case that $M=\lambda y.M'$, and $N=\lambda y.N'$, where $M'$ and $N'$ again satisfy assumptions of the lemma.
		The derivation of $\Gamma\vdash N:(u\oplus_\ell w,\tau)$ ends with the $(\lambda)$ rule,
		above which we have a type judgment $\Gamma'\vdash N':(u\oplus_\ell w,\tau')$, to which we apply the induction hypothesis.
		We obtain a derivation of $\Gamma'\vdash M':(u'\oplus_\ell w',\tau')$, where $u',w'$ satisfy the required conditions;
		to this type judgment we apply again the $(\lambda)$ rule, deriving $\Gamma\vdash M:(u'\oplus_\ell w',\tau)$ as required.
		
		Finally, it may be the case that $M=P\,Q$ and $N=R\,S$, where either $P=R$ and the reduction changes $Q$ to $S$, or $Q=S$ and the reduction changes $P$ to $R$.
		The derivation of $\Gamma\vdash N:(u\oplus_\ell w,\tau)$ ends with the $(@)$ rule,
		above which, for some set of indices $I$, we have some type judgments $\Gamma_i\vdash N_i:(u_i\oplus_\ell w_i,\tau_i)$ for $i\in I$.
		We can write
		\begin{align}\label{eq:uw}
			u=D_\mathsf{u}+\sum_{i\in I}u_i&&\mbox{and}&&
			w=D_\mathsf{w}+\sum_{i\in I}w_i,
		\end{align}
		where $D_\mathsf{u},D_\mathsf{w}$ are appropriate duplication factors.
		The lambda-terms $N_i$ are either $R$ or $S$; let us define $M_i$ to be $P$ if $N_i=R$ and $S$ if $N_i=S$ (where $i\in I$).
		For each of the above type judgments we either use the induction hypothesis, or the fact that $M_i=N_i$ (in which case we simply take the same type judgment)
		to obtain a derivation of $\Gamma_i\vdash M_i:(u_i'\oplus_\ell w_i',\tau_i)$ for productivity values $u_i', w_i'$ such that, for all $a\in\setP$,
		\begin{align}
			&2^{u_i(a)}\cdot w_i(a)\leq 2^{u_i'(a)}\cdot w_i'(a),\label[inequality]{eq:11}\\
			&u_i(a)\leq u_i'(a),&&\mbox{and}\label[inequality]{eq:12}\\
			&u_i(a)+w_i(a)=0\Rightarrow u_i'(a)+w_i'(a)=0.\label[implication]{eq:13}
		\end{align}
		Note that \cref{eq:13} can be actually changed into an equivalence: the right-to-left implication follows from \cref{eq:11,eq:12}.
		Such an equivalence allows us to apply the $(@)$ rule again, and derive $\Gamma\vdash M:(u'\oplus_\ell w',\tau)$, where
		\begin{align}\label{eq:upwp}
			u'=D_\mathsf{u}+\sum_{i\in I}u_i'&&\mbox{and}&&w'=D_\mathsf{w}+\sum_{i\in I}w_i'
		\end{align}
		for the same duplication factors $D_\mathsf{u},D_\mathsf{w}$ as in the original derivation.
		In order to prove the required conditions productivity values, fix some $a\in\setP$.
		After expanding $u,w,u',w'$ using \cref{eq:uw,eq:upwp},
		the first inequality that we need to prove becomes
		\begin{align}\label[inequality]{eq:1-goal}
			2^{D_\mathsf{u}(a)+\sum_{i\in I}u_i(a)}\cdot\Big(D_\mathsf{w}(a)+\sum_{i\in I}w_i(a)\Big)
			\leq 2^{D_\mathsf{u}(a)+\sum_{i\in I}u_i'(a)}\cdot\Big(D_\mathsf{w}(a)+\sum_{i\in I}w_i'(a)\Big).
		\end{align}
		Using \cref{eq:12} for every $i\in I$ we obtain that
		\begin{align}\label[inequality]{eq:1-11}
			2^{D_\mathsf{u}(a)+\sum_{i\in I}u_i(a)}\cdot D_\mathsf{w}(a)
			\leq 2^{D_\mathsf{u}(a)+\sum_{i\in I}u_i'(a)}\cdot D_\mathsf{w}(a).
		\end{align}
		Moreover, using \cref{eq:11} for a single $i\in I$ and \cref{eq:12} for all $i'\in I$ other than $i$, we obtain that
		\begin{align}\label[inequality]{eq:1-12}
			2^{D_\mathsf{u}(a)+\sum_{i'\in I}u_i(a)}\cdot w_i(a)
			\leq 2^{D_\mathsf{u}(a)+\sum_{i'\in I}u_i'(a)}\cdot w_i'(a).
		\end{align}
		Summing \cref{eq:1-12} over all $i\in I$ and adding \cref{eq:1-11} we obtain \cref{eq:1-goal}, as required.
		The second inequality to be proven, namely $u(a)\leq u'(a)$, follows directly from \cref{eq:uw,eq:upwp,eq:12}.
		Likewise, the implication $u(a)+w(a)=0\Rightarrow u'(a)+w'(a)=0$ follows directly from \cref{eq:uw,eq:upwp,eq:13}.
	\end{proof}

\subsection{Ensuring homogeneity}\label{app:homogeneous}

	As already said, we conduct our proof of \cref{lem:completeness-fin} using homogeneous lambda-terms.
	In order to justify this, we now prove that homogeneity can be indeed ensured for safe lambda-terms:

	\begin{lemma}~\label{lem:completeness-hom}
		Every finite closed safe lambda-term $M$ of sort $\otyp$ and complexity $m$ can be converted
		into a finite closed safe homogeneous lambda-term $M'$ of sort $\otyp$ and complexity $m$
		such that $\BT(M')=\BT(M)$ and if we can derive $\emptyset\vdash M':(v,\r)$, then we can derive $\emptyset\vdash M:(v,\r)$.
	\end{lemma}

	The conversion is performed by following ideas from Parys' work~\cite{homogeneous}.
	Namely, first all the subterms of $M$ are put in $\eta$-long form (which is formalized below as elimination of \emph{problematic proper subterms}),
	and then we reorder parameters and arguments to make the new term homogeneous.
	
	First, let us underline that in this subsection we work with finite lambda-terms.
	Second, in this subsection we operate on occurrences of subterms;
	we do not give a completely formal definition, but while talking about an occurrence of a subterm $N$ in a lambda-term $M$ we mean a particular location (address) where $N$ occurs in $M$.
	
	We say that an occurrence of a subterm $N$ in a lambda-term $M$ is \emph{improper} if
	\begin{itemize}
	\item	$N$ is an application (i.e., $N=K\,L$) and this occurrence of $N$ is on the operand position of a longer application $N\,L'$, or
	\item	$N$ starts with a lambda-binder (i.e., $N=\lambda x.K$) and this occurrence of $N$ is surrounded by another lambda-binder (i.e., we have $\lambda y.N$).
	\end{itemize}
	Otherwise, the occurrence of $N$ in $M$ is \emph{proper}.
	We say that a proper occurrence of $N$ in $M$ is \emph{problematic} if $N$ is not of sort $\otyp$ and either
	\begin{itemize}
	\item	$N$ is an application, or
	\item	this occurrence of $N$ is surrounded by a lambda-binder (i.e., we have $\lambda x.N$).
	\end{itemize}
	When no proper occurrence of any subterm $N$ in $M$ is problematic, we say that $M$ \emph{has no problematic subterms}.

	Having the above definitions, we can now state requirements for the first transformation, eliminating problematic proper subterms:

	\begin{lemma}~\label{lem:eta-long}
		Every finite closed safe lambda-term $M$ of sort $\otyp$ and complexity $m$ can be converted
		into a finite closed safe lambda-term $M'$ of sort $\otyp$ and complexity $m$ that has no problematic proper subterms
		and such that $\BT(M')=\BT(M)$ and if we can derive $\emptyset\vdash M':(v,\r)$, then we can derive $\emptyset\vdash M:(v,\r)$.
	\end{lemma}

	\begin{proof}
		Let $N$ be a lambda-term of sort $\alpha_1\arr\dots\arr\alpha_k\arr\otyp$, where $k\geq 1$.
		We define $\floor{N}$ to be
		\begin{align*}
			\lambda x_1\lamdots\lambda x_k.(N\,x_1\,\dots\,x_k),
		\end{align*}
		where each $x_i$ is a fresh variable of sort $\alpha_i$.
		Note that $N$ and $\floor{N}$ have the same sort, and that every derivation of $\Gamma\vdash\floor{N}:(u,\tau)$ contains a subderivation for $\Gamma\vdash N:(u,\tau)$.
		Indeed, say that $\tau=\bigwedge_{j\in J_1}(A_{1j},\tau_{1j})\arr\dots\arr\bigwedge_{j\in J_k}(A_{kj},\tau_{kj})\arr\r$.
		The derivation of $\Gamma\vdash\floor{N}:(u,\tau)$ has to end with $k$ applications of the $(\lambda)$ rule, above which we have
		\begin{align*}
			\Gamma\cup\textstyle\bigcup_{i=1}^k \setof{x_i: (A_{ij},\tau_{ij})}{j\in J_i }\vdash N\,x_1\,\dots\,x_k:(u,\r).
		\end{align*}
		Above this type judgment we have, in turn, $k$ applications of the $(@)$ rule
		with type judgments $\set{x_i:(A_{ij},\tau_{ij})}\vdash x_i:(\zero,\tau_{ij})$ for the arguments $x_1,\dots,x_k$ and with $\Gamma\vdash N:(u,\tau)$.
		Since the variables are all fresh (they do not occur in $N$, so they cannot be present in $\Gamma$),
		the duplication factor is zero at each step and the value is not changed.
		
		We repair $M$ by repeating the following process:
		take some problematic proper occurrence of a subterm $N$ and replace it by $\floor{N}$.

		First, note that such a replacement preserves safety:
		\begin{itemize}
		\item	The lambda-terms $N$ and $\floor{N}$ have the same free variables,
			so every superficially safe lambda-term remains superficially safe after replacing the considered occurrence of $N$ by $\floor{N}$.
			It is thus enough to check safety ``near the border of $N$''.
		\item	Possibly the considered occurrence of $N$ is on the operand position in some application $N\,L$.
			In such a case, in the new lambda-term we have $\floor{N}\,L$, and $\floor{N}$ is not an application, so it is required that $\floor{N}$ is superficially safe.
			But since the considered occurrence is proper, $N$ cannot be an application.
			This means that $N$ is superficially safe, and so is $\floor{N}$.
		\item	Inside $\floor{N}$ we have a new application $N\,x_1\,\dots\,x_k$.
			The subterms $x_1,\dots,x_k$ are superficially safe by definition.
			If $N$ is not an application, it is also required that $N$ is superficially safe.
			By definition of a problematic occurrence, if $N$ is not an application, then it is surrounded by a lambda-binder;
			because the occurrence is proper, $N$ cannot start with a lambda-binder;
			it is either a variable or a constant, hence it is superficially safe by definition.
		\end{itemize}
		
		Second, note that, after the replacement, the number of problematic proper occurrences of subterms indeed decreases:
		\begin{itemize}
		\item	The occurrence of $N$ inside $\floor{N}$ is on the operand position of an application $N\,x_1$ (in particular, it is not surrounded by a lambda-binder);
			if this occurrence is proper, then $N$ is not an application, so the occurrence is not problematic.
		\item	The subterm $N\,x_1\,\dots\,x_k$ is of sort $\otyp$, so its occurrence (being proper) is not problematic.
		\item	The whole $\floor{N}$ starts with a lambda-binder (in particular $\floor{N}$ is not an application);
			if the occurrence of $\floor{N}$ is proper, it is not surrounded by a lambda-binder, so it is not problematic.
		\end{itemize}
		
		Moreover, clearly the resulting lambda-term remains closed, of sort $\otyp$, and of complexity $m$.

		At the end we obtain a lambda-term $M'$ without any problematic proper subterms
		and based on a derivation of $\emptyset\vdash M':(v,\r)$ we can get a derivation of $\emptyset\vdash M:(v,\r)$.
		Moreover, $M'$ is eta-convertible to $M$, and thus their beta-normal forms are also eta-convertible.
		But since $M$ and $M'$ are closed lambda-terms of sort $\otyp$, their beta-normal forms contain no variables, meaning that these beta-normal forms are equal,
		that is, $\BT(M')=\BT(M)$.
	\end{proof}
	
	\begin{proof}[Proof of \cref{lem:completeness-hom}]
		By \cref{lem:eta-long} we can assume that our lambda-term has no problematic proper subterms.
		Now we can go to the second step: we rearrange the arguments of each proper subterm in such a way that it becomes homogeneous.
		To this end, we define a function $\sort$ on sorts, types, and lambda-terms.

		Consider a sort $\alpha=\alpha_1\arr\dots\arr\alpha_k\arr\otyp$.
		Proceeding by induction on the structure of $\alpha$ we define $\sort(\alpha)$ to be $\sort(\alpha_{s_1})\arr\dots\arr\sort(\alpha_{s_k})\arr\otyp$,
		where $s_1,\dots,s_k$ is the (unique) permutation of $1,\dots,k$ such that for all $i,j\in\set{1,\dots,k}$ with $i<j$
		we have $\ord(\alpha_{s_i})\geq\ord(\alpha_{s_j})$ and if $\ord(\alpha_{s_i})=\ord(\alpha_{s_j})$ then $s_i<s_j$.
		In other words, $\sort(\alpha)$ is obtained by sorting $\alpha_1,\dots,\alpha_k$ by order (from high to low, preserving the ordering of sorts having the same order)
		and then modifying these sub-sorts recursively.
		Note that $\ord(\sort(\alpha))=\ord(\alpha)$.
		
		Likewise for a type $\tau=(\bigwedge T_1\arr\dots\arr\bigwedge T_k\arr\r)\in\Tt^\alpha_s$ we define $\sort(\tau)$ to be
		\begin{align*}
			\left(\bigwedge\setof{\sort(\sigma)}{\sigma\in T_{s_1}}\arr\dots\arr\bigwedge\setof{\sort(\sigma)}{\sigma\in T_{s_k}}\arr\r\right)\in\Tt^{\sort(\alpha)}_s,
		\end{align*}
		where $s_1,\dots,s_k$ is as above.
		
		Based on that, we modify lambda-terms: for a proper subterm $N$ of $M$ we define $\sort(N)$ to be
		\begin{itemize}
		\item	$N$ if $N$ is a variable or a constant,
		\item	$\sort(K)\,\sort(L_{s_1})\,\dots\,\sort(L_{s_k})$ if $N=K\,L_1\,\dots\,L_k$ for $K$ not being an application and $k\geq 1$,
			where $s_1,\dots,s_k$ are as above assuming that $K$ is of sort $\alpha$, and
		\item	$\lambda x_{s_1}\lamdots\lambda x_{s_k}.\sort(K)$ if $N=\lambda x_1\lamdots\lambda x_k.K$ for $K$ not starting with a lambda-binder and $k\geq 1$,
			where $s_1,\dots,s_k$ are as above assuming that $N$ is of sort $\alpha$.
		\end{itemize}
		For $N$ of sort $\alpha$ the resulting lambda-term $\sort(N)$ is of sort $\sort(\alpha)$,
		and has the same free variables as $N$.
		
		In the above definition, it is important that there are no problematic proper subterms in $M$.
		Thanks to that, in the case of an application, the lambda-term $N$ is necessarily of sort $\otyp$,
		meaning that we provide to $K$ all $k$ arguments required according to its sort,
		and we can indeed reorder these arguments.
		Likewise, in the case of a lambda-abstraction, the lambda-term $K$ is necessarily of sort $\otyp$,
		meaning that we surround $K$ by lambda-binders for all $k$ parameters of $N$,
		and we can indeed reorder these lambda-binders.

		At the end, we take $M'=\sort(M)$.
		It is immediate from the construction that $M'$ is finite, closed, safe, of sort $\otyp$, of complexity $m$, and moreover homogeneous.

		It is also true that $\BT(M')=\BT(M)$.
		Intuitively, this is the case because whenever we reorder some lambda-binders in $M$, we also reorder corresponding arguments,
		so the correspondence between parameter names and arguments is preserved.
		We omit here a detailed formal proof, which is contained in the Parys' work~\cite{homogeneous} for a similar setting.
		
		Finally, we need to observe that if we can derive $\setof{x:(A,\sort(\tau))}{(x:(A,\tau))\in\Gamma}\vdash N:(v,\sort(\tau))$,
		then we can also derive $\Gamma\vdash N:(v,\tau)$.
		A new derivation is obtained by appropriately modifying the original one.
		Namely, whenever we have a sequence of $(@)$ rules for a compound application, we simply reorder them;
		notice that the total duplication factor added by these rules remains unchanged.
		Likewise, we reorder every sequence of $(\lambda)$ rules for a compound lambda-abstraction.
	\end{proof}

\subsection{Proof of Lemma~\ref{lem:completeness-fin}}

	We first have a \lcnamecref{lem:completeness-fin-stop} describing lambda-terms in the beta-normal form:

	\begin{lemma}\label{lem:completeness-fin-stop}
		If $M$ is a closed lambda-term of sort $\otyp$ and complexity $0$,
		and $v\colon\setP\to\Nat$ is a function such that in $\BT(M)$ there is a finite branch having, for every $a\in\setP$, exactly $v(a)$ occurrences of $a$,
		then we can derive $\emptyset\vdash M:(v,\r)$.
	\end{lemma}

	\begin{proof}
		We prove the \lcnamecref{lem:completeness-fin-stop} by induction on the length of the selected branch in $\BT(M)$.
		Fix a closed lambda-term $M$ of sort $\otyp$ and complexity $0$, and a finite branch in $\BT(M)$;
		let $v\colon\setP\to\Nat$ be such that, for every $a\in\setP$, on the fixed branch there are exactly $v(a)$ occurrences of $a$.
		Note that $M$ (being closed, of sort $\otyp$, and of complexity $0$) is necessarily of the form either $\leta_j\,M'$, or $\letb\,M_1\,M_2$, or $\letc$, or $\upomega$:
		\begin{itemize}
		\item	If $M=\leta_j\,M'$ (and $\BT(M)=\leta_j\,\BT(M')$), the branch visits the $\leta_j$-labeled root and continues in $\BT(M')$.
			By the induction hypothesis we can derive $\emptyset\vdash M':(v-\chi_j,\r)$.
			Then we derive $\emptyset\vdash\leta_j:(\chi_j,(A,\r)\arr\r)$ for $A=\setof{a\in\setP}{(v-\chi_j)(a)>0}$ using the rule for $\leta_j$,
			and we conclude with $\emptyset\vdash M:(v,\r)$ after applying the $(@)$ rule.
		\item	If $M=\letb\,M_1\,M_2$ (and $\BT(M)=\letb\,\BT(M_1)\,\BT(M_2)$), the branch visits the $\letb$-labeled root and continues in $\BT(M_i)$ for some $i\in\set{1,2}$.
			By the induction hypothesis we can derive $\emptyset\vdash M_i:(v,\r)$.
			Using an appropriate rule for $\letb$ (involving either the first, or the second argument---depending on $i$) and the $(@)$ rule twice, we again obtain
			$\emptyset\vdash M:(v,\r)$.
			Note that we do not need to derive any type judgment concerning the second subterm, $M_{3-i}$.
		\item	For $M=\letc$, the branch ends in the $\letc$-labeled root of $\BT(M)=\letc$, so $v=\zero$.
			We derive $\emptyset\vdash M:(v,\r)$ using the rule for $\letc$.
		\item	It is impossible that $M=\upomega$, since a branch, by definition, cannot visit an $\upomega$-labeled node.
		\qedhere\end{itemize}
	\end{proof}
	
	In the next \lcnamecref{lem:find-good} we need an additional definition (occurring also in prior work~\cite{schemes-lY,blum-ong-safe}):
	a lambda-term $M$ is \emph{almost safe} if for its every subterm of the form $K\,L_1\,\dots\,L_k$, where $K$ is not an application and $k\geq 1$,
	all subterms $K,L_1,\dots,L_k$ are superficially safe
	(the difference, when compared to the definition of a safe lambda-term, is that we do not require the whole $M$ to be superficially safe).
	
	In \cref{lem:find-good} we show how to find a redex for which one can use \cref{lem:compl-1red-multi-term}:
	
	\begin{lemma}\label{lem:find-good}
		Let $\ell\geq 0$, and let $M$ be a finite, homogeneous, almost safe lambda-term of complexity $\ell+1$ and order at most $\ell$,
		with all free variables of order at most $\ell-1$.
		In such a situation, $M$ has a subterm of the form $(\lambda x.K)\,L$, where $L$ is closed, has order $\ell$, and does not use any variables of order at least $\ell$.
	\end{lemma}

	\begin{proof}
		The proof is by induction on the size of $M$.
		Note that every subterm of a finite, almost safe, homogeneous lambda-term is again finite, almost safe, and homogeneous.
		Moreover, $M$ cannot be a variable nor a constant ($M$ should be of order at most $\ell$, but should have a subterm of order $\ell+1$).
		
		If $M=\lambda y.P$, then $\ord(P)\leq\ord(M)\leq\ell$ and $\ord(y)\leq\ord(M)-1\leq\ell-1$ (all free variables of $P$ are again of order at most $\ell-1$);
		we can conclude using the induction hypothesis for $P$ (its complexity is again $\ell+1)$.
		
		It remains to consider the case when $M$ is an application;
		let $M=P\,Q_1\,Q_2\,\dots\,Q_k$, where $P$ is not an application and $k\geq 1$.
		Of course all free variables of the subterms are again of order at most $\ell-1$.
		Moreover, $\ord(Q_i)\leq\ell$ for each $i\in\set{1,\dots,k}$:
		if $P$ is a constant then clearly $\ord(Q_i)=0\leq\ell$, and otherwise the order of $P$ is included to the complexity of $M$, so $\ord(Q_i)\leq\ord(P)-1\leq\ell+1-1$.
		If, for some $i\in\set{1,\dots,k}$, the complexity of $Q_i$ is $\ell+1$, then we can use the induction hypothesis for $Q_i$.
		Thus suppose conversely: for each $i\in\set{1,\dots,k}$ the complexity of $Q_i$ is at most $\ell$.
		This implies that the complexity of $P$ is $\ell+1$ (a subterm causing that the complexity of $M$ is $\ell+1$ is either located inside $P$,
		or it is $P\,Q_1\,\dots\,Q_i$ for some $i\in\set{1,\dots,k}$, but in the latter case also $P$ itself has order $\ell+1$,
		which is included to the complexity of $P$).
		If $\ord(P)\leq\ell$, we can use the induction hypothesis for $P$, and we are done.
		We can thus assume that $\ord(P)=\ell+1$.
		Note that $P$ is not a constant (by definition the order of a constant is not included to its complexity;
		the complexity of a constant is $0$, while the complexity of $P$ is $\ell+1\geq 1$),
		nor a variable ($M$ has no free variables of order $\ell+1$),
		so $P$ is of the form $\lambda x.K$.
		By homogeneity, $\ord(Q_1)=\ell$ (the first argument of $P$ has the highest order).
		Because $M$ is almost safe, $Q_1$ is superficially safe, so all its free variables are of order at least $\ord(Q_1)$, that is, at least $\ell$.
		But, by assumption, all free variables of $M$ (hence also of $Q_1$) are of order at most $\ell-1$,
		which implies that $Q_1$ is closed.
		Observe also that $Q_1$ does not use any variables of order at least $\ell$.
		Indeed, every variable $y$ used in $Q_1$ is introduced in a subterm of the form $\lambda y.S$ with $\ord(\lambda y.S)\geq\ord(y)+1$;
		the order of $\lambda y.S$ does not exceed the complexity of $Q_1$, which is at most $\ell$, so $\ord(y)\leq\ell-1$.
		We thus have the thesis of the \lcnamecref{lem:find-good} with $L=Q_1$.
	\end{proof}

	Next, we observe that the reductions we consider preserve safety (we remark that this is not true for arbitrary beta-reductions):
	
	\begin{lemma}\label{lem:preserve-safety}
		Let $M,N$ be lambda-terms such that $N$ is obtained from $M$ by reducing a subterm of the form $(\lambda x.K)\,L$ in which $L$ is closed.
		If $M$ is safe, then $N$ is safe as well.
	\end{lemma}
	
	\begin{proof}
		The lambda-terms $(\lambda x.K)\,L$ and $K[L/x]$ have the same free variables, so every superficially safe lambda-term containing the occurrence of $(\lambda x.K)\,L$ remains
		superficially safe after the reduction.
		Likewise, every subterm of $K$ that is superficially safe remains superficially safe after substituting $L$ for $x$ in it,
		because this substitution does not add any new free variables (here we use the assumption that $L$ is closed).
		Next, safety of $N$ may require that $L$ is superficially safe (if $L$ is substituted into some application), but $L$ in $M$ is an argument of an application,
		so it is indeed superficially safe.
		Finally, it may be also needed that $K[L/x]$ is superficially safe.
		But we know that $\lambda x.K$ is superficially safe, that $\lambda x.K$ and $K[L/x]$ have the same free variables (here we use again the fact that $L$ is closed),
		and $\ord(K[L/x])\leq\ord(\lambda x.K)$; this implies that indeed $K[L/x]$ is superficially safe, and finishes the proof.
	\end{proof}
	
	Having the above lemmata, we can now restate \cref{lem:compl-1red-multi-term} in the style of \cref{lem:soundness-red-multi}:
	
	\begin{lemma}\label{lem:completeness-step}
		If $M$ is a finite, homogeneous, safe, closed lambda-term of sort $\otyp$ and complexity $\ell+1$, where $\ell\geq 0$,
		then there exists a safe lambda-term $N$ such that $M\redb N$
		and such that whenever we can derive $\emptyset\vdash N:(u\oplus_\ell w,\r)$ then there exist $u',w'$ such that we can derive $\emptyset\vdash M:(u'\oplus_\ell w',\r)$
		and $2^{u(a)}\cdot w(a)\leq 2^{u'(a)}\cdot w'(a)$.
	\end{lemma}
	
	\begin{proof}
		First, using \cref{lem:find-good} we find in $M$ a redex $(\lambda x.K)\,L$ such that $L$ is closed, has order $\ell$, and does not use any variables of order at least $\ell$.
		As $N$ we take a lambda-term in which some occurrence of this redex is reduced, that is, replaced by $K[L/x]$.
		Clearly $M\redb N$.
		Moreover, by \cref{lem:preserve-safety}, the resulting lambda-term is safe.
		Finally, a derivation of $\emptyset\vdash N:(u\oplus_\ell w,\r)$ can be converted into a derivation of $\emptyset\vdash M:(u'\oplus_\ell w',\r)$ for appropriate $u',w'$
		by \cref{lem:compl-1red-multi-term}.
	\end{proof}

	Finally, we prove the following \lcnamecref{lem:completeness-fin-multi}, being a multi-letter counterpart of \cref{lem:completeness-fin}:

	\begin{lemma}\label{lem:completeness-fin-multi}
		For every $m\in\Nat$ there exists a function $H_m\colon\Nat\to\Nat$ such that if $M$ is a closed safe lambda-term of sort $\otyp$ and complexity at most $m$,
		and $v\colon\setP\to\Nat$ is a function such that in $\BT(M)$ there is a finite branch having, for every $a\in\setP$, at least $v(a)$ occurrences of $a$,
		then there exists $v'$ such that we can derive $\emptyset\vdash M:(v',\r)$ and $v(a)\leq H_m(v'(a))$ for all $a\in\setP$.
	\end{lemma}

	\begin{proof}
		We define $H_m$ as a tower of powers of $2$ of height $m$;
		formally, $H_0(n)=n$ and $H_m(n)=2^{H_{m-1}(n)}$ for $m\geq 1$.
		Fix now a closed safe lambda-term $M$ of sort $\otyp$ and complexity at most $m$, fix a finite branch in $\BT(M)$,
		and fix a function $v\colon\setP\to\Nat$ such that, for every $a\in\setP$, on the fixed branch there are at least $v(a)$ occurrences of $a$.
		
		Suppose first that $M$ is finite and homogeneous.
		In this case we prove the thesis by induction on $m$.
		If $m=0$, as $v'$ we take the function such that, for every $a\in\setP$, on the fixed branch there are exactly $v'(a)$ occurrences of $a$;
		clearly $v(a)\leq v'(a)=H_m(v'(a))$, so the thesis follows from \cref{lem:completeness-fin-stop}.
		
		For $m\geq 1$ we define $\ell=m-1$ and we prove a stronger thesis:
		there exist $u,w$ such that we can derive $\emptyset\vdash M:(u\oplus_\ell w,\r)$ and $v(a)\leq H_\ell(2^{u(a)}\cdot w(a))$ for all $a\in\setP$.
		This thesis is indeed stronger, since by \cref{lem:std-eqv-ext}(a) we can then derive $\emptyset\vdash M:(u+w,\r)$,
		and $H_\ell(2^{u(a)}\cdot w(a))\leq H_\ell(2^{u(a)}\cdot 2^{w(a)})=H_m(u(a)+w(a))$.
		We prove this stronger thesis by an internal induction on the maximal number $u$ such that there exists a sequence of $n$ beta-reductions from $M$
		(recall hat $M$, hence also this number, is finite).
		One possibility is that $M$ is actually of complexity at most $\ell$ (this includes the base case of $n=0$).
		Then we simply use the induction hypothesis of the external induction saying that we can derive $\emptyset\vdash M:(v',\r)$
		with $v(a)\leq H_\ell(v'(a))=H_\ell(2^{\zero(a)}\cdot v'(a))$ for all $a\in\setP$;
		\cref{lem:std-eqv-ext}(b) implies that we can then derive $\emptyset\vdash M:(\zero\oplus_\ell v',\r)$, as needed.
		
		The other possibility is that $M$ has complexity exactly $m$.
		We then apply \cref{lem:completeness-step}, which gives us a safe lambda-term $N$ such that $M\redb N$.
		Clearly $N$ is again finite, closed, homogeneous, of sort $\otyp$, and of complexity at most $m$,
		and the maximal length of a sequence of beta-reductions starting in $N$ is smaller than for $M$.
		We can thus apply the induction hypothesis (of the internal induction) to $N$, and obtain a derivation of $\emptyset\vdash N:(u\oplus_\ell w,\r)$,
		where $v(a)\leq H_\ell(2^{u(a)}\cdot w(a))$ for all $a\in\setP$.
		Based on it, from \cref{lem:completeness-step} we obtain a derivation of $\emptyset\vdash M:(u'\oplus_\ell w',\r)$
		where $2^{u(a)}\cdot w(a)\leq 2^{u'(a)}\cdot w'(a)$
		(hence also $H_\ell(2^{u(a)}\cdot w(a))\leq H_\ell(2^{u'(a)}\cdot w'(a))$) for all $a\in\setP$, as needed.
		This finishes the proof in the case when $M$ is finite and homogeneous.

		Next, suppose that $M$ is finite, but not homogeneous.
		In this case we use \cref{lem:completeness-hom} to obtain a finite closed safe homogeneous lambda-term $M'$ of sort $\otyp$, of complexity at most $m$,
		and such that $\BT(M')=\BT(M)$.
		By the already considered case we obtain a derivation of $\emptyset\vdash M':(v',\r)$ with $v(a)\leq H_m(v'(a))$ for all $a\in\setP$,
		and \cref{lem:completeness-hom} gives us a derivation of $\emptyset\vdash M:(v',\r)$, as required.

		Finally, suppose that $M$ is infinite.
		Recall the definition of a cut, introduces in the proof of \cref{lem:soundness-fin-multi}.
		It is known from prior work~\cite[Lemma 4.2]{diagonal-types} that having a lambda-term $M$ and a finite branch of $\BT(M)$ one can find a finite cut $M'$ of $M$
		such that the same branch exists also in $\BT(M')$.
		From the previous part of our proof we obtain a derivation of $\emptyset\vdash M':(v',\r)$ with $v(a)\leq H_m(v'(a))$ for all $a\in\setP$.
		It is easy to change it into a derivation of $\emptyset\vdash M:(v',\r)$.
		Indeed, in order to obtain $M$ from $M'$ one has to replace some subterms of the form $\lambda x_1\lamdots\lambda x_k.\upomega$ with some other lambda-terms.
		We can replace the corresponding subterms also in lambda-terms occurring in all type judgments of our derivation.
		The key point is that in the derivation we have cannot have a type judgment concerning a subterm of the form $\lambda x_1\lamdots\lambda x_k.\upomega$
		(because there is no rule for $\upomega$), so after replacing those subterms the derivation is still correct.
	\end{proof}

\end{document}